\documentclass[pdflatex,sn-mathphys,Numbered]{sn-jnl}

\usepackage{graphicx}%
\usepackage{multirow}%
\usepackage{amsmath,amssymb,amsfonts}%
\usepackage{amsthm}%
\usepackage{mathrsfs}%
\usepackage[title]{appendix}%
\usepackage{xcolor}%
\usepackage{textcomp}%
\usepackage{manyfoot}%
\usepackage{booktabs}%
\usepackage{algorithm}%
\usepackage{algorithmicx}%
\usepackage{algpseudocode}%
\usepackage{listings}%

\usepackage{quiver}
\usepackage{mathtools}
\usepackage{todonotes}
\newcommand{\tob}[1]{\todo[color=green!30,inline,caption={}]{\textbf{T:} #1}}
\newcommand{\fab}[1]{\todo[color=red!30,inline,caption={}]{\textbf{F:} #1}}
\newcommand{\dav}[1]{\todo[color=blue!30,inline,caption={}]{\textbf{D:} #1}}
\newcommand{\andrea}[1]{\todo[color=yellow!30,inline,caption={}]{\textbf{A:} #1}}
\newcommand{\comment}[1]{ }
\usepackage{hyperref}
\usepackage{xcolor}

\theoremstyle{definition}
\newtheorem*{notation*}{Notation}

\newcommand{\freccia}[3]{#2 \colon #1  \to #3}

\newcommand{\duefreccia}[3]{\xymatrix@C=0.5cm{#2 \colon #1  \ar@{=>}[r] &  #3}}
\newcommand{\duefreccianoname}[2]{#1\leq #2}

\newcommand{\comsquare}[8]{ \xymatrix@+1pc{ 
#1 \ar[r]^{#5} \ar[d]_{#6} & #2 \ar[d]^{#7} \\
#3 \ar[r]_{#8} & #4 
}}
\newcommand{\pullback}[8]{ \xymatrix@+1pc{ 
#1 \pullbackcorner \ar[r]^{#5} \ar[d]_{#6} & #2 \ar[d]^{#7} \\
#3 \ar[r]_{#8} & #4 
}}
\newcommand{\quadratocomm}[8]{ \xymatrix@+1pc{ 
#1 \ar[r]^{#5} \ar[d]_{#6} & #2 \ar[d]^{#7} \\
#3 \ar[r]_{#8} & #4 
}}
\newcommand{\comsquarelargo}[8]{ \xymatrix@+1pc{ 
#1 \ar[rr]^{#5} \ar[d]_{#6} && #2 \ar[d]^{#7} \\
#3 \ar[rr]_{#8} && #4 
}}
\newcommand{\parallelmorphisms}[4]{\xymatrix@+1pc{
#1 \ar @<+4pt>[r]^{#2} \ar @<-4pt>[r]_{#3} & #4
}}
\newcommand{\relation}[4]{\xymatrix@+1pc{
\angbr{#2}{#3}\colon #1 \ar @<+4pt>[r] \ar @<-4pt>[r] & #4
}}
\newcommand{\frecceparalleleopposte}[4]{\xymatrix@+1pc{
#1 \ar@<+4pt>[r]^{#2} \ar@<-4pt>@{<-}[r]_{#3} & #4
}}
\newcommand{\equalizer}[6]{\xymatrix@+1pc{
#1 \ar[r]^{#2} & #3 \ar @<+4pt>[r]^{#4} \ar @<-4pt>[r]_{#5} & #6
}}
\newcommand{\coequalizer}[6]{\xymatrix@+1pc{
 #1 \ar @<+4pt>[r]^{#2} \ar @<-4pt>[r]_{#3} & #4 \ar[r]^{#5} & #6
}}

\newcommand{\subobject}[3]{\xymatrix{
#1 \ar@{>->}[r]^{#2} & #3
}}

\newcommand{\pullbackcorner}[1][ul]{\save*!/#1+1.2pc/#1:(1,-1)@^{|-}\restore}

\def\mA{\mathcal{A}}

\def\mC{\mathcal{C}}

\def\mD{\mathcal{D}}

\def\FinSet{\mathbf{FinSet}}
\def\FinRel{\mathbf{FinRel}}
\def\Set{\mathbf{Set}}
\def\Rel{\mathbf{Rel}}
\def\Preord{\mathbf{Preord}}
\newcommand{\pfn}[1]{#1\mbox{-}\mathbf{Fun}}
\newcommand{\fn}[1]{#1\mbox{-}\mathbf{TFun}}
\newcommand{\total}[1]{#1\mbox{-}\mathbf{Total}}

\def\ox{\otimes}
\def\id{\operatorname{ id}}         
\def\op{\operatorname{ op}}         
\def\dom{\operatorname{ dom}}      
\def\I{I}                          

\newcommand{\angbr}[2]{\langle #1,#2 \rangle}

\usepackage{tikz}
\usepackage{tikzit}

\usetikzlibrary{shapes}

\usepackage{freetikz}
\usetikzlibrary{decorations.markings,positioning,patterns}
\usetikzlibrary{shadows}
\usepackage{array}
\usepackage{tikz-cd}

\tikzstyle{nodonero}=[fill=black, draw=black, shape=circle]
\tikzstyle{box}=[fill=white, draw=black, shape=rectangle]

\usepackage[all,2cell]{xy}
\UseAllTwocells
\xyoption{v2}
\theoremstyle{plain} 
\newtheorem{mytheorem}{Theorem}[section]
\newtheorem{mycorollary}[mytheorem]{Corollary}
\newtheorem{mylemma}[mytheorem]{Lemma}
\newtheorem{myproposition}[mytheorem]{Proposition}
\newtheorem{mydefinition}[mytheorem]{Definition}

\theoremstyle{definition} 
\newtheorem{myproblem}[mytheorem]{Problem}
\newtheorem{myremark}[mytheorem]{Remark}
\newtheorem{myexample}[mytheorem]{Example}

\newcommand{\stringdiagnabla}[1]{\begin{tikzpicture}[scale=0.60, transform shape]
	\begin{pgfonlayer}{nodelayer}
		\node [style=none] (0) at (-4.25, 0) {};
		\node [style=none] (1) at (-2.75, 0.5) {};
		\node [style=none] (2) at (-2.75, -0.5) {};
		\node [style=none] (3) at (-2, -0.5) {};
		\node [style=none] (4) at (-2, 0.5) {};
		\node [style=nodonero] (5) at (-3.25, 0) {};
		\node [style=none] (7) at (-3.75, 0.25) {#1};
	\end{pgfonlayer}
	\begin{pgfonlayer}{edgelayer}
		\draw (0.center) to (5);
		\draw [bend left, looseness=1.25] (5) to (1.center);
		\draw (1.center) to (4.center);
		\draw [bend right] (5) to (2.center);
		\draw (2.center) to (3.center);
	\end{pgfonlayer}
	\pgfsetbaseline{-0.4ex}
\end{tikzpicture}}

\newcommand{\stringdiagbang}[1]{\begin{tikzpicture}[scale=0.60, transform shape]
	\begin{pgfonlayer}{nodelayer}
		\node [style=none] (8) at (2, 0) {};
		\node [style=nodonero] (9) at (3, 0) {};
		\node [style=none] (10) at (2.5, 0.25) {#1};
	\end{pgfonlayer}
	\begin{pgfonlayer}{edgelayer}
		\draw (8.center) to (9);
	\end{pgfonlayer}
		\pgfsetbaseline{-.4ex}
\end{tikzpicture}}

\newcommand{\stringdiagfreccia}[3]{\begin{tikzpicture}[scale=0.60, transform shape]
	\begin{pgfonlayer}{nodelayer}
		\node [style=box] (5) at (1.75, 0) {#2};
		\node [style=none] (7) at (3, 0) {};
		\node [style=none] (8) at (2.5, 0.25) {#3};
		\node [style=none] (9) at (0.5, 0) {};
		\node [style=none] (10) at (1, 0.25) {#1};
	\end{pgfonlayer}
	\begin{pgfonlayer}{edgelayer}
		\draw (5) to (7);
		\draw (9.center) to (5);
	\end{pgfonlayer}
	\pgfsetbaseline{-.4ex}
\end{tikzpicture}
}

\begin{document}

\title[From gs-monoidal to oplax cartesian categories]{From gs-monoidal to oplax cartesian categories:\\ constructions and functorial completeness}


\author[1]{\fnm{Tobias} \sur{Fritz}}\email{tobias.fritz@uibk.ac.at}
\equalcont{These authors contributed equally to this work.}
\author[2]{\fnm{Fabio} \sur{Gadducci}}\email{fabio.gadducci@unipi.it}
\equalcont{These authors contributed equally to this work.}
\author*[2]{\fnm{Davide} \sur{Trotta}}\email{trottadavide92@gmail.com}
\equalcont{These authors contributed equally to this work.}
\author[2]{\fnm{Andrea} \sur{Corradini}}\email{andrea.corradini@unipi.it}
\equalcont{These authors contributed equally to this work.}
\affil[1]{\orgdiv{Department of Mathematics}, \orgname{University of Innsbruck}, \orgaddress{ \city{Innsbruck}, \country{Austria}}}

\affil[2]{\orgdiv{Department of Computer Science}, \orgname{University of Pisa}, \orgaddress{\city{Pisa}, \country{Italy}}}


\abstract{Originally introduced in the context of the algebraic approach to term graph rewriting, 
    the notion of gs-monoidal category has surfaced a few times under different monikers 
    in the last decades.
    They can be thought of as symmetric monoidal categories whose arrows are generalised relations, with enough structure
    to talk about domains and partial functions, but less structure than cartesian bicategories.
    %
    The aim of this paper is threefold. The first goal is to extend the original definition
    of gs-monoidality by enriching it with a preorder on arrows,
    giving rise to what we call oplax cartesian categories.
    %
    Second, we show that (preorder-enriched) gs-monoidal categories naturally arise both as Kleisli categories and as
    span categories, and the relation between the resulting formalisms is explored.
    %
    Finally, 
    we present two theorems concerning Yoneda embeddings on the one hand
    and functorial completeness on the other, the latter inducing a
    completeness result also for lax functors from oplax cartesian categories to $\Rel$.}




\maketitle
\setcounter{tocdepth}{1}
\tableofcontents

\section{Introduction}\label{sec1}

The notion of \emph{gs-monoidal category} was originally introduced in the context of the algebraic approach to term graph rewriting~\cite{gadducci1996,CorradiniGadducci97}.
Their study was pursued in a series of papers (see e.g. \cite{CorradiniGadducci99,CorradiniGadducci99b} among others), including their application to the functorial semantics of relational and partial algebras~\cite{CorradiniGadducci02}.
%

Briefly, gs-monoidal categories are symmetric monoidal categories equipped with two families of arrows, a duplicator $\freccia A {\nabla_A} {A \otimes A}$ and a
discharger $\freccia A {!_A} I$ for each object $A$, subject to a few coherence axioms. 
These arrows are not required to be natural in $A$, and this fact is precisely what accounts for the difference between considering terms as trees or as graphs.
In fact, it is often observed that if naturality holds, then the monoidal product 
is the categorical product, and thus the category is cartesian monoidal~\cite{Fox:CACC}, a structure used in Lawvere theories~\cite{Lawvere869} to 
represent abstractly \emph{functional} algebraic operations and their compositions.
From this perspective, the lack of naturality of duplicators and dischargers in gs-monoidal categories leads one to think of their arrows as generalised \emph{relations} and \emph{partial functions} instead.

Conceptually, gs-monoidal categories are a weaker version of cartesian bicategories~\cite{CARBONI198711,cartesianbicatII},
lacking the dual arrows for duplicators and dischargers. The relevance of cartesian bicategories to mathematics and computer science increased in the last years, see e.g.~\cite{bonchi_seeber_sobocinski_18,Bonchi2017c,Fong19}, and the notion of gs-monoidality has surfaced a few times under different monikers.
The simplicity of this notion and its pervasiveness has led several 
authors to investigate such structures, in some cases independently developed.
This is the case for the work by Golubtsov~\cite{Gol99}, whose 
\emph{categories of information transformers} are very similar to gs-monoidal categories.
More recently, gs-monoidal categories appeared as
\emph{copy-discard categories}~\cite{cho_jacobs_2019}.
Their \emph{affine} variant is the basis for
a recent approach to categorical probability, where these categories are dubbed
 \emph{Markov categories}~\cite{Fritz_2020} based on the interpretation of the arrows
 as generalised Markov kernels.

There is a conceptual hierarchy of categories, from symmetric monoidal to cartesian ones, sketched in the following diagram, with forgetful functors going upwards
{\small
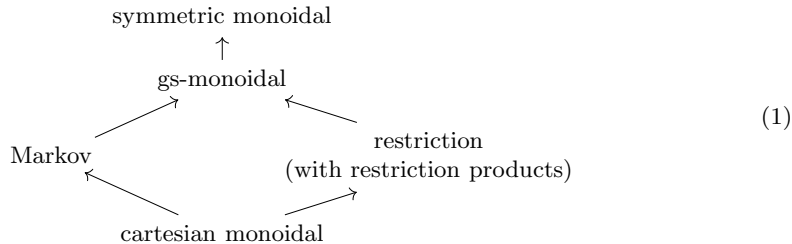
		\begin{equation}
			\label{cats_diagram}
			\begin{tikzcd}[row sep=9pt, column sep=0pt]
				& \mbox{symmetric monoidal} \ar[d,<-] & \\
				& \mbox{gs-monoidal} \ar[dl,<-] \ar[dr,<-]&\\ 
				\mbox{Markov}  \ar[dr,<-]& & \hspace{-1cm}\parbox{4cm}{\centering restriction\\ (with restriction products)} \ar[dl,<-]\\
				& \mbox{cartesian monoidal}  &				
			\end{tikzcd}
		\end{equation}
}
Indeed as already mentioned, a \emph{gs-monoidal} category is a symmetric monoidal category with 
 a duplicator ${\nabla_A}$ and a discharger ${!_A}$ arrow for each object $A$, subject to a few coherence axioms but not
 to naturality; \emph{Markov} categories are gs-monoidal categories where the discharger is natural; and dually
 \emph{restriction categories with restriction products}~\cite{Cockett07} are precisely those where the duplicator is natural,
as explained later in Section~\ref{sec:on gs}. Finally, if both dischargers and duplicators are natural, the monoidal product
is the categorical product and thus we get \emph{cartesian monoidal} categories.


Motivated by the current interest in these categorical structures, in this work we aim to explore more in depth the original notion in several directions. 
We first provide an overview of the main characteristics of gs-monoidality and of its preorder-enriched version, which culminates in the introduction of \emph{oplax cartesian categories}.
Our presentation adopts the graphical formalism of string diagrams, and it highlights the main properties 
underlying gs-monoidal categories.
%
This also facilitates a simple treatment of the connection with a well-established proposal for the categorical modelling of partiality,
restriction categories \cite{Cockett02,Cockett03,Cockett07}. 

 
  
 We then explore  two different settings in which the gs-monoidal and oplax cartesian structures naturally arise,
 namely, Kleisli categories and span categories.
Similarly to what was observed in \cite{Fritz_2020} in the context of Markov categories and affine monads, the Kleisli category
of a commutative monad on a gs-monoidal category $\mA$ is shown to be gs-monoidal.
Moreover, for a suitable notion of gs-monoidal monad $T$ on an oplax cartesian 
category $\mA$, we show that the canonical functor $\mA_T \to \mA$ preserves the oplax cartesian 
structure.

Moving on, we generalise an almost folklore result, namely, that the category $\mathbf{PSpan}(\mA)$, 
obtained from the bicategory of spans of a category with finite limits $\mA$ by identifying arrows whenever 
they are isomorphic as spans and by taking the preorder reflection of the 2-cells, is gs-monoidal and 
admits a canonical oplax cartesian structure. 
We thus obtain a gs-monoidal functor between Kleisli and span categories.
%

Finally, we turn our attention to functorial completeness questions for preorder-enriched gs-monoidal categories: which types of functors are needed in order to distinguish any two distinct morphisms?
Our first step in this direction is to formulate and prove a gs-monoidal Yoneda embedding. Our second one
is to obtain a functorial completeness theorem with respect to suitable bilax functors~\cite{aguiar2010}, 
which induces a completeness theorem
also for functors of oplax cartesian categories into $\Rel$
and thus offers a tool for exploiting gs-monoidal and oplax cartesian categories in the setting 
of functorial semantics for relational and partial algebras, much in the spirit of~\cite{CorradiniGadducci02}.

\subsubsection*{Overview}

Section 2 recalls the basic notions of gs-monoidal  categories and overviews their main properties.
Section 3 introduces preorder-enriched gs-monoidal categories and oplax cartesian categories, setting the stage for
our later results.
Section 4 presents a characterisation of Kleisli and span categories having the gs-monoidal and oplax cartesian structure, respectively,
and provides a formal link between those categories via an oplax cartesian functor.
Section 5 shows a gs-monoidal Yoneda embedding and a completeness result for oplax cartesian categories. 
Section 6 discusses future work.

The reader may find some categorical background in Appendices~\ref{sec:lax_app}--\ref{sec: strong and commutative monad}.

\section{Background on gs-monoidal categories}
\label{sec:on gs}


Originally introduced in the context of algebraic approaches to term graph rewriting~\cite{gadducci1996,CorradiniGadducci97}, the notion of 
\emph{gs-monoidal category} has been 
developed in a series of papers \cite{CorradiniGadducci99, CorradiniGadducci02, CorradiniGadducci99b}.
We recall here the basic definition adopting the graphical formalism of string diagrams, referring to \cite{Selinger2011} for an overview of various notions of monoidal categories and their associated diagrammatic calculus.

\begin{mydefinition}[gs-monoidal category]\label{def gs-monoidal cat}
A \textbf{gs-monoidal category} $\mC$ is a symmetric monoidal category, where we denote by $\otimes$ the tensor product and by $\I$ the unit, such that every object $A$ of $\mC$ is equipped with arrows 
$$\stringdiagnabla{A}:A \to A\ox A \qquad \stringdiagbang{A}: A\to I$$
which satisfy the commutative comonoid axioms
\begin{equation}
\label{comm_comon}
\begin{split}
\begin{tikzpicture}[scale=0.60, transform shape]
	\begin{pgfonlayer}{nodelayer}
		\node [style=none] (0) at (-5.5, 1.25) {};
		\node [style=none] (1) at (-4, 1.75) {};
		\node [style=none] (2) at (-4, 0.75) {};
		\node [style=none] (3) at (-3.25, 0.75) {};
		\node [style=none] (4) at (-3.25, 1.75) {};
		\node [style=nodonero] (5) at (-4.5, 1.25) {};
		\node [style=none] (7) at (-5, 1.5) {A};
		\node [style=none] (8) at (-3.25, 1.75) {};
		\node [style=none] (9) at (-1.75, 2.25) {};
		\node [style=none] (10) at (-1.75, 1.25) {};
		\node [style=none] (11) at (-1, 1.25) {};
		\node [style=none] (12) at (-1, 2.25) {};
		\node [style=nodonero] (13) at (-2.25, 1.75) {};
		\node [style=none] (14) at (0, 1.25) {};
		\node [style=none] (15) at (1.5, 1.75) {};
		\node [style=none] (16) at (1.5, 0.75) {};
		\node [style=none] (17) at (2.25, 0.75) {};
		\node [style=none] (18) at (2.25, 1.75) {};
		\node [style=nodonero] (19) at (1, 1.25) {};
		\node [style=none] (20) at (0.5, 1.5) {A};
		\node [style=none] (21) at (-1, 0.75) {};
		\node [style=none] (23) at (-0.5, 1.25) {$=$};
		\node [style=none] (24) at (2.25, 0.75) {};
		\node [style=none] (25) at (3.75, 1.25) {};
		\node [style=none] (26) at (3.75, 0.25) {};
		\node [style=none] (27) at (4.5, 0.25) {};
		\node [style=none] (28) at (4.5, 1.25) {};
		\node [style=nodonero] (29) at (3.25, 0.75) {};
		\node [style=none] (30) at (4.5, 1.75) {};
		\node [style=none] (49) at (-7.25, -1.5) {};
		\node [style=none] (50) at (-7.25, -1.5) {};
		\node [style=none] (51) at (-5.75, -1) {};
		\node [style=none] (52) at (-5.75, -2) {};
		\node [style=nodonero] (55) at (-6.25, -1.5) {};
		\node [style=none] (56) at (-4.25, -1) {};
		\node [style=none] (57) at (-4.25, -2) {};
		\node [style=none] (58) at (-3.25, -1.5) {};
		\node [style=none] (59) at (-1.75, -1) {};
		\node [style=none] (60) at (-1.75, -2) {};
		\node [style=none] (61) at (-1, -2) {};
		\node [style=none] (62) at (-1, -1) {};
		\node [style=nodonero] (63) at (-2.25, -1.5) {};
		\node [style=none] (64) at (-2.75, -1.25) {A};
		\node [style=none] (65) at (-3.75, -1.5) {$=$};
		\node [style=none] (66) at (-1, -2) {};
		\node [style=none] (73) at (-6.75, -1.25) {A};
		\node [style=none] (87) at (1.25, -1.5) {};
		\node [style=none] (88) at (2.75, -1) {};
		\node [style=none] (89) at (2.75, -2) {};
		\node [style=none] (90) at (3.5, -2) {};
		\node [style=none] (91) at (3.5, -1) {};
		\node [style=nodonero] (92) at (2.25, -1.5) {};
		\node [style=none] (93) at (1.75, -1.25) {A};
		\node [style=none] (94) at (3.5, -2) {};
		\node [style=nodonero] (95) at (3.5, -1) {};
		\node [style=none] (96) at (4, -1.5) {$=$};
		\node [style=none] (97) at (4.5, -1.5) {};
		\node [style=none] (98) at (5.5, -1.5) {};
		\node [style=none] (99) at (5, -1.25) {A};
	\end{pgfonlayer}
	\begin{pgfonlayer}{edgelayer}
		\draw (0.center) to (5);
		\draw [bend left, looseness=1.25] (5) to (1.center);
		\draw (1.center) to (4.center);
		\draw [bend right] (5) to (2.center);
		\draw (2.center) to (3.center);
		\draw (8.center) to (13);
		\draw [bend left, looseness=1.25] (13) to (9.center);
		\draw (9.center) to (12.center);
		\draw [bend right] (13) to (10.center);
		\draw (10.center) to (11.center);
		\draw (14.center) to (19);
		\draw [bend left, looseness=1.25] (19) to (15.center);
		\draw (15.center) to (18.center);
		\draw [bend right] (19) to (16.center);
		\draw (16.center) to (17.center);
		\draw (3.center) to (21.center);
		\draw (24.center) to (29);
		\draw [bend left, looseness=1.25] (29) to (25.center);
		\draw (25.center) to (28.center);
		\draw [bend right] (29) to (26.center);
		\draw (26.center) to (27.center);
		\draw (18.center) to (30.center);
		\draw (50.center) to (55);
		\draw [bend left, looseness=1.25] (55) to (51.center);
		\draw [bend right] (55) to (52.center);
		\draw [in=-180, out=0] (52.center) to (56.center);
		\draw [in=180, out=0, looseness=0.75] (51.center) to (57.center);
		\draw (58.center) to (63);
		\draw [bend left, looseness=1.25] (63) to (59.center);
		\draw (59.center) to (62.center);
		\draw [bend right] (63) to (60.center);
		\draw (60.center) to (61.center);
		\draw (87.center) to (92);
		\draw [bend left, looseness=1.25] (92) to (88.center);
		\draw (88.center) to (91.center);
		\draw [bend right] (92) to (89.center);
		\draw (89.center) to (90.center);
		\draw (97.center) to (98.center);
	\end{pgfonlayer}
\end{tikzpicture}
\end{split}
\end{equation}
and the monoidal multiplicativity axioms
\begin{equation}
\begin{split}
\label{monoidal_mult}
\begin{tikzpicture}[scale=0.60, transform shape]
	\begin{pgfonlayer}{nodelayer}
		\node [style=none] (0) at (-7, 0) {};
		\node [style=none] (1) at (-5.5, 0.5) {};
		\node [style=none] (2) at (-5.5, -0.5) {};
		\node [style=none] (3) at (-4.75, -0.5) {};
		\node [style=none] (4) at (-4.75, 0.5) {};
		\node [style=nodonero] (5) at (-6, 0) {};
		\node [style=none] (6) at (-6.75, 0.25) {$A\otimes B$};
		\node [style=none] (9) at (-4, 0) {=};
		\node [style=none] (10) at (-3.25, 1) {};
		\node [style=none] (11) at (-1.75, 1.5) {};
		\node [style=none] (12) at (-1.75, 0.5) {};
		\node [style=none] (13) at (0, -0.5) {};
		\node [style=none] (14) at (0, 1.5) {};
		\node [style=nodonero] (15) at (-2.25, 1) {};
		\node [style=none] (16) at (-2.75, 1.25) {A};
		\node [style=none] (18) at (-3.25, -1) {};
		\node [style=none] (19) at (-1.75, -0.5) {};
		\node [style=none] (20) at (-1.75, -1.5) {};
		\node [style=none] (21) at (0, -1.5) {};
		\node [style=none] (22) at (0, 0.5) {};
		\node [style=nodonero] (23) at (-2.25, -1) {};
		\node [style=none] (24) at (-2.75, -0.75) {B};
		\node [style=none] (25) at (-6, -4) {};
		\node [style=nodonero] (26) at (-5, -4) {};
		\node [style=none] (27) at (-5.75, -3.75) {$A\otimes B$};
		\node [style=none] (29) at (-4, -4) {=};
		\node [style=none] (30) at (-3.25, -3.5) {};
		\node [style=nodonero] (31) at (-2.25, -3.5) {};
		\node [style=none] (32) at (-2.75, -3.25) {A};
		\node [style=none] (33) at (-3.25, -4.5) {};
		\node [style=nodonero] (34) at (-2.25, -4.5) {};
		\node [style=none] (35) at (-2.75, -4.25) {B};
		\node [style=none] (37) at (0, -1.25) {B};
		\node [style=none] (38) at (0, 0.75) {B};
		\node [style=none] (39) at (0, -0.25) {A};
		\node [style=none] (40) at (0, 1.75) {A};
		\node [style=none] (41) at (5.5, -4) {$=$};
		\node [style=none] (42) at (6, -4) {};
		\node [style=none] (43) at (7, -4) {};
		\node [style=none] (44) at (6.5, -3.75) {I};
		\node [style=none] (45) at (3.75, -4) {};
		\node [style=nodonero] (46) at (4.75, -4) {};
		\node [style=none] (47) at (4.25, -3.75) {I};
		\node [style=none] (48) at (2.75, 0) {};
		\node [style=none] (49) at (4.25, 0.5) {};
		\node [style=none] (50) at (4.25, -0.5) {};
		\node [style=none] (51) at (5, -0.5) {};
		\node [style=none] (52) at (5, 0.5) {};
		\node [style=nodonero] (53) at (3.75, 0) {};
		\node [style=none] (54) at (3, 0.25) {$I$};
		\node [style=none] (55) at (5.5, 0) {=};
		\node [style=none] (56) at (6, 0) {};
		\node [style=none] (57) at (7, 0) {};
		\node [style=none] (58) at (6.5, 0.25) {I};
	\end{pgfonlayer}
	\begin{pgfonlayer}{edgelayer}
		\draw (0.center) to (5);
		\draw [bend left, looseness=1.25] (5) to (1.center);
		\draw (1.center) to (4.center);
		\draw [bend right] (5) to (2.center);
		\draw (2.center) to (3.center);
		\draw (10.center) to (15);
		\draw [bend left, looseness=1.25] (15) to (11.center);
		\draw (11.center) to (14.center);
		\draw [bend right] (15) to (12.center);
		\draw [in=180, out=0, looseness=1.25] (12.center) to (13.center);
		\draw (18.center) to (23);
		\draw [bend left, looseness=1.25] (23) to (19.center);
		\draw [in=-180, out=0] (19.center) to (22.center);
		\draw [bend right] (23) to (20.center);
		\draw (20.center) to (21.center);
		\draw (25.center) to (26);
		\draw (30.center) to (31);
		\draw (33.center) to (34);
		\draw (42.center) to (43.center);
		\draw (45.center) to (46);
		\draw (48.center) to (53);
		\draw [bend left, looseness=1.25] (53) to (49.center);
		\draw (49.center) to (52.center);
		\draw [bend right] (53) to (50.center);
		\draw (50.center) to (51.center);
		\draw (56.center) to (57.center);
	\end{pgfonlayer}
\end{tikzpicture}
\end{split}
\end{equation}
\end{mydefinition}
Symbolically, we also write $\nabla_A : A \to A \otimes A$ for the first structure arrow above and call it \textbf{duplicator}, and similarly $!_A : A \to I$ for the 
\textbf{discharger}.
\begin{myremark}
	The monoidal multiplicativity equations come in two pairs, one pair specifying the duplicator and discharger on a tensor product object and one for doing so on the monoidal unit $I$. However, the two equations for $I$ imply each other upon using the counitality axiom (bottom-right of \eqref{comm_comon}), thus one of them could be omitted.
\end{myremark}
\begin{myremark}
Following the style of presentation of string diagrams, the axioms are given for \emph{strict} monoidal categories, i.e. where 
the coherence isomorphisms for associativity and unitality are identities. This is without loss of generality, either by considering
the strictification
of the categories at hand~\cite[Proposition~3.28]{fongspivak2019} or by adding the required coherence isomorphisms to the axioms.
\end{myremark}

\begin{myexample}
	\label{comon_is_gs}
	If $\mC$ is the category of commutative comonoids in a symmetric monoidal category $\mA$, with arrows given by the arrows of $\mA$ without any further conditions, then $\mC$ is a gs-monoidal category in a canonical way: duplicators and dischargers are given by the comultiplications and counits of the comonoids, respectively.
	Then the commutative comonoid equations~\eqref{comm_comon} hold by definition of commutative comonoid, and the monoidal multiplicativity equations~\eqref{monoidal_mult} hold by definition of the monoidal product of commutative comonoids.
\end{myexample}

%
As for functors between symmetric monoidal categories, also functors between gs-monoidal categories come in several variants. The definitions that follow refer to lax, oplax, strong and strict monoidal functors (see e.g.~\cite{aguiar2010}) recalled in Appendix~\ref{sec:lax_app}. 
\begin{mydefinition}\label{def gs monoidal functor}
	For gs-monoidal categories $\mC$ and $\mD$, a functor $\freccia{\mC}{F}{\mD}$ is
	\begin{enumerate}
		\item \textbf{lax gs-monoidal} if it is equipped with a lax symmetric monoidal structure  
			\[
				\freccia{\otimes \circ \, (F\times F)}{\psi}{F\circ \otimes}, \qquad \freccia{I}{\psi_0}{F(I)} 
			\]
such that the following diagrams commute 
        \[\begin{tikzcd}[column sep=tiny]
            {F(A)} && {F(A\otimes A)}   &&& FA && {F(I)} \\
            & {F(A)\otimes F(A)} &&&&& I
            \arrow["{F(\nabla_A)}", from=1-1, to=1-3]
            \arrow["{\nabla_{FA}}"', from=1-1, to=2-2]
            \arrow["{\psi_{A,A}}"', from=2-2, to=1-3]
            \arrow[ "{F(!_A)}", from=1-6, to=1-8]
             \arrow["{!_{FA}}"', from=1-6, to=2-7]
            \arrow["{\psi_0}"', from=2-7, to=1-8]
       \end{tikzcd}\]
		\item \textbf{oplax gs-monoidal} if it is equipped with an oplax symmetric monoidal structure 
		\[
			\freccia{F\circ \otimes }{\phi}{ \otimes\circ (F\times F)}, \qquad \freccia{F(I)}{\phi_0}{I}
		\]
such that the following diagrams commute 
		\[\begin{tikzcd}[column sep=tiny]
            {F(A)} && {F(A\otimes A)}   &&& FA && {F(I)} \\
            & {F(A)\otimes F(A)} &&&&& I
            \arrow["{F(\nabla_A)}", from=1-1, to=1-3]
            \arrow["{\nabla_{FA}}"', from=1-1, to=2-2]
            \arrow["{\phi_{A,A}}", from=1-3, to=2-2]
            \arrow["{F(!_A)}", from=1-6, to=1-8]
            \arrow["{!_{FA}}"', from=1-6, to=2-7]
            \arrow["{\phi_0}", from=1-8, to=2-7]
          \end{tikzcd}\]
		\item \textbf{strong gs-monoidal} if it is strong symmetric monoidal and the above diagrams commute;
		\item \textbf{strict gs-monoidal} if it is strict symmetric monoidal, thus in particular 
		$F(A \otimes B) = F(A) \otimes F(B)$ and $F(I) = I$, and satisfies
			\[
				F(\nabla_A) = \nabla_{F(A)}, \qquad F(!_A) = {}!_{F(A)}.
			\]
	\end{enumerate}
\end{mydefinition}

As in the purely monoidal case, ``gs-monoidal functor'' without further qualification refers to the strong version.

Recall also the notion of \emph{bilax symmetric monoidal} functor from Definition~\ref{def:bilax monoidalfunctor}. We then obtain the definition of \textbf{bilax gs-monoidal functor} upon adding the four commutative triangles above to that definition.

\subsection{Total and functional arrows, domains}

In this section, we investigate the categorical structure of functional and total arrows. 
This allows us to explore the connection between gs-monoidal categories, Markov categories, restriction categories \cite{Cockett02}, 
p-categories \cite{Robinson88}, and cartesian monoidal categories, as depicted in~\eqref{cats_diagram}.

\begin{mydefinition}\label{def: C-fun and C-tot}
Let $\mC$ be a gs-monoidal category. An arrow $\freccia{A}{f}{B}$ is
$\mC$-\textbf{total} if
\[
\begin{tikzpicture}[scale=0.60, transform shape]
	\begin{pgfonlayer}{nodelayer}
		\node [style=none] (0) at (3.75, 0) {=};
		\node [style=none] (1) at (4.5, 0) {};
		\node [style=nodonero] (2) at (5.5, 0) {};
		\node [style=none] (3) at (5, 0.25) {A};
		\node [style=box] (5) at (1.75, 0) {$f$};
		\node [style=nodonero] (7) at (3, 0) {};
		\node [style=none] (8) at (2.5, 0.25) {B};
		\node [style=none] (9) at (0.5, 0) {};
		\node [style=none] (10) at (1, 0.25) {A};
	\end{pgfonlayer}
	\begin{pgfonlayer}{edgelayer}
		\draw (1.center) to (2);
		\draw (5) to (7);
		\draw (9.center) to (5);
	\end{pgfonlayer}
\end{tikzpicture}
\]
and $\mC$-\textbf{functional} if
\[
\begin{tikzpicture}[scale=0.60, transform shape]
	\begin{pgfonlayer}{nodelayer}
		\node [style=none] (0) at (-5.75, 0) {=};
		\node [style=box] (11) at (-8.75, 0) {$f$};
		\node [style=none] (13) at (-8, 0.25) {B};
		\node [style=none] (14) at (-10, 0) {};
		\node [style=none] (15) at (-9.5, 0.25) {A};
		\node [style=none] (17) at (-7, 0.5) {};
		\node [style=none] (18) at (-7, -0.5) {};
		\node [style=none] (19) at (-6.25, -0.5) {};
		\node [style=none] (20) at (-6.25, 0.5) {};
		\node [style=nodonero] (21) at (-7.5, 0) {};
		\node [style=none] (23) at (-6.25, -0.5) {};
		\node [style=none] (24) at (-5.25, 0) {};
		\node [style=none] (25) at (-3.75, 0.5) {};
		\node [style=none] (26) at (-3.75, -0.5) {};
		\node [style=nodonero] (29) at (-4.25, 0) {};
		\node [style=none] (30) at (-4.75, 0.25) {A};
		\node [style=box] (32) at (-2.5, -0.5) {$f$};
		\node [style=none] (34) at (-1.75, -0.75) {B};
		\node [style=none] (35) at (-3.75, -0.5) {};
		\node [style=none] (36) at (-3.25, -0.75) {};
		\node [style=none] (37) at (-1.25, -0.5) {};
		\node [style=box] (39) at (-2.5, 0.5) {$f$};
		\node [style=none] (40) at (-1.75, 0.75) {B};
		\node [style=none] (42) at (-3.25, 0.75) {};
		\node [style=none] (43) at (-1.25, 0.5) {};
	\end{pgfonlayer}
	\begin{pgfonlayer}{edgelayer}
		\draw (14.center) to (11);
		\draw [bend left, looseness=1.25] (21) to (17.center);
		\draw (17.center) to (20.center);
		\draw [bend right] (21) to (18.center);
		\draw (18.center) to (19.center);
		\draw (11) to (21);
		\draw (24.center) to (29);
		\draw [bend left, looseness=1.25] (29) to (25.center);
		\draw [bend right] (29) to (26.center);
		\draw (35.center) to (32);
		\draw (32) to (37.center);
		\draw (39) to (43.center);
		\draw (25.center) to (39);
	\end{pgfonlayer}
\end{tikzpicture}
\]
We denote the subcategory of $\mC$-functional arrows by $\pfn{\mC}$, the one of $\mC$-total arrows by 
$\total{\mC}$, and the one of $\mC$-total and $\mC$-functional arrows by $\fn{\mC}$.
\end{mydefinition}

\begin{myexample}\label{comon2}
	If $\mC$ is the gs-monoidal category of commutative comonoids in a symmetric monoidal category $\mA$, i.e. the full subcategory of $\mA$ whose objects are the commutative comonoids of $\mA$ (see Example~\ref{comon_is_gs}), then the total arrows are the counital ones, the functional arrows the comultiplicative ones, and therefore $\fn{\mC}$ is the category of commutative comonoids and comonoid homomorphisms in $\mA$.
\end{myexample}

\begin{myexample}\label{ex: rel}
The category $\Rel$ of sets and relations with the monoidal structure $\freccia{\Rel\times \Rel}{\otimes}{\Rel}$ given by the cartesian product of sets is gs-monoidal~\cite{CorradiniGadducci02}. In particular, $\Rel$-functional arrows are precisely partial functions, $\Rel$-total arrows are precisely total relations, and $\Rel$-total, $\Rel$-functional arrows are precisely functions.
\end{myexample}

\begin{myexample}\label{ex: Markov categories}
	Every Markov category $\mathcal{M}$ in the sense of~\cite{Fritz_2020,FritzGPR23} is gs-monoidal. In fact, Markov categories are exactly those gs-monoidal categories whose monoidal unit $I$ is terminal, or equivalently those for which every arrow is $\mathcal{M}$-total.
	The $\mathcal{M}$-functional arrows in the sense of Definition \ref{def: C-fun and C-tot} are precisely those called \textit{deterministic} there. 
\end{myexample}

We note a useful property of $\mC$-total and $\mC$-functional arrows, 
which generalise corresponding observations for Markov 
categories~\cite[Lemma~10.12]{Fritz_2020}.

\begin{myproposition}\label{prop: ! and nabla are total and functional}
Let $\mC$ be a gs-monoidal category. Then
   $\pfn{\mC}$ and $\total{\mC}$ are gs-monoidal subcategories of $\mC$. 
\end{myproposition}

\begin{proof}
Closure under composition and $\otimes$ is straightforward.
To finish the proof, we show that both $!$ and $\nabla$ are $\mC$-total and $\mC$-functional: every arrow $!_A$ is $\mC$-total because $!_{\I}!_A=\id_I!_A$ since $!_{\I}=\id_{\I}$ by  Definition \ref{def gs-monoidal cat}. Moreover $!_A$ is $\mC$-functional because

\[
\begin{tikzpicture}[scale=0.60, transform shape]
	\begin{pgfonlayer}{nodelayer}
		\node [style=none] (0) at (-5.75, 0) {=};
		\node [style=none] (14) at (-9.5, 0) {};
		\node [style=none] (15) at (-8.5, 0.25) {A};
		\node [style=none] (17) at (-7, 0.5) {};
		\node [style=none] (18) at (-7, -0.5) {};
		\node [style=none] (19) at (-6.25, -0.5) {};
		\node [style=none] (20) at (-6.25, 0.5) {};
		\node [style=nodonero] (21) at (-7.5, 0) {};
		\node [style=none] (23) at (-6.25, -0.5) {};
		\node [style=none] (24) at (-5.25, 0) {};
		\node [style=nodonero] (29) at (-4, 0) {};
		\node [style=none] (30) at (-4.75, 0.25) {A};
		\node [style=nodonero] (44) at (-6.25, 0.5) {};
		\node [style=nodonero] (45) at (-6.25, -0.5) {};
	\end{pgfonlayer}
	\begin{pgfonlayer}{edgelayer}
		\draw [bend left, looseness=1.25] (21) to (17.center);
		\draw (17.center) to (20.center);
		\draw [bend right] (21) to (18.center);
		\draw (18.center) to (19.center);
		\draw (24.center) to (29);
		\draw (14.center) to (21);
	\end{pgfonlayer}
\end{tikzpicture}
\]
by strictness of $\mC$ and by the axioms of Definition \ref{def gs-monoidal cat}.
The same calculation also shows that every duplicator $\nabla_A$ is $\mC$-total.
It is $\mC$-functional by the first monoidal multiplicativity axiom combined with the commutative comonoid equations.
\end{proof}

We recall a few almost folklore properties (see e.g.\cite{Fox:CACC, Gol99, Bruni2003}) needed later on.

\begin{myproposition}\label{thm C gs-cat implies C-total maps have weak cart prod}\label{thm total functions have products}
Let $\mC$ be a gs-monoidal category. 
Then
\begin{enumerate}
	\item the subcategory $\total{\mC}$ of $\mC$-total arrows has weak binary products given by $\otimes$;
	\item the subcategory $\fn{\mC}$ of $\mC$-functional and $\mC$-total arrows is cartesian monoidal.
\end{enumerate}
\end{myproposition}

\begin{proof}
 	\begin{enumerate}
		\item Let us consider two objects $A$ and $B$ of $\mC$ and show that $A\ox B$ is a weak product in $\total{\mC}$ with respect to the projections
			$\freccia{A\ox B}{\id_A\ox \,!_B}{A}$ and $ \freccia{A\ox B}{!_A \ox{} \id_B}{B}$, i.e.
	\[
	\begin{tikzpicture}[scale=0.60, transform shape]
	\begin{pgfonlayer}{nodelayer}
		\node [style=none] (14) at (-9.25, 0) {};
		\node [style=none] (15) at (-8.25, -0.5) {B};
		\node [style=nodonero] (21) at (-7.75, 0) {};
		\node [style=none] (22) at (-9.25, 1) {};
		\node [style=none] (23) at (-7.25, 1) {};
		\node [style=none] (24) at (-8.25, 1.5) {A};
		\node [style=none] (25) at (-3.75, 0) {};
		\node [style=none] (26) at (-2.75, -0.5) {B};
		\node [style=none] (28) at (-3.75, 1) {};
		\node [style=none] (29) at (-2.25, 1) {};
		\node [style=none] (30) at (-2.75, 1.5) {A};
		\node [style=nodonero] (31) at (-2.25, 1) {};
		\node [style=none] (32) at (-1.75, 0) {};
	\end{pgfonlayer}
	\begin{pgfonlayer}{edgelayer}
		\draw (14.center) to (21);
		\draw (22.center) to (23.center);
		\draw (28.center) to (29.center);
		\draw (25.center) to (32.center);
	\end{pgfonlayer}
\end{tikzpicture}
\]
By Proposition \ref{prop: ! and nabla are total and functional}, these projections are $\mC$-total arrows, since they are monoidal products of $\mC$-total arrows. 

Now let us consider two $\mC$-total arrows $\freccia{C}{h}{A}$ and $\freccia{C}{g}{B}$. First, observe that the arrow
\[
\begin{tikzpicture}[scale=0.60, transform shape]
	\begin{pgfonlayer}{nodelayer}
		\node [style=none] (14) at (-9.25, 0) {};
		\node [style=none] (15) at (-8.25, 0.25) {C};
		\node [style=none] (17) at (-6.75, 0.5) {};
		\node [style=none] (18) at (-6.75, -0.5) {};
		\node [style=none] (19) at (-6, -0.5) {};
		\node [style=none] (20) at (-6, 0.5) {};
		\node [style=nodonero] (21) at (-7.25, 0) {};
		\node [style=none] (23) at (-6, -0.5) {};
		\node [style=box] (33) at (-6, 0.5) {$h$};
		\node [style=box] (34) at (-6, -0.5) {$g$};
		\node [style=none] (36) at (-4.75, 0.5) {};
		\node [style=none] (37) at (-5.5, 1) {};
		\node [style=none] (38) at (-5.5, -1) {};
		\node [style=none] (39) at (-5.5, -1) {B};
		\node [style=none] (40) at (-5.5, 1) {A};
		\node [style=none] (60) at (-4.75, -0.5) {};
	\end{pgfonlayer}
	\begin{pgfonlayer}{edgelayer}
		\draw [bend left, looseness=1.25] (21) to (17.center);
		\draw (17.center) to (20.center);
		\draw [bend right] (21) to (18.center);
		\draw (18.center) to (19.center);
		\draw (14.center) to (21);
		\draw (33) to (36.center);
		\draw (34) to (60.center);
	\end{pgfonlayer}
\end{tikzpicture}
\]
is $\mC$-total by Proposition \ref{prop: ! and nabla are total and functional}. To conclude the statement, it is enough to show that
\[
\begin{tikzpicture}[scale=0.60, transform shape]
	\begin{pgfonlayer}{nodelayer}
		\node [style=none] (0) at (-4, 0) {=};
		\node [style=none] (14) at (-9.25, 0) {};
		\node [style=none] (15) at (-8.25, 0.25) {C};
		\node [style=none] (17) at (-6.75, 0.5) {};
		\node [style=none] (18) at (-6.75, -0.5) {};
		\node [style=none] (19) at (-6, -0.5) {};
		\node [style=none] (20) at (-6, 0.5) {};
		\node [style=nodonero] (21) at (-7.25, 0) {};
		\node [style=none] (23) at (-6, -0.5) {};
		\node [style=none] (24) at (-3.25, 0) {};
		\node [style=none] (30) at (-2.75, 0.25) {C};
		\node [style=box] (31) at (-1.75, 0) {$h$};
		\node [style=none] (32) at (0, 0) {};
		\node [style=box] (33) at (-6, 0.5) {$h$};
		\node [style=box] (34) at (-6, -0.5) {$g$};
		\node [style=nodonero] (35) at (-5.25, -0.5) {};
		\node [style=none] (36) at (-4.75, 0.5) {};
		\node [style=none] (37) at (-5.5, 1) {};
		\node [style=none] (38) at (-5.5, -1) {};
		\node [style=none] (39) at (-5.5, -1) {B};
		\node [style=none] (40) at (-5.5, 1) {A};
		\node [style=none] (41) at (-0.75, 0.25) {};
		\node [style=none] (42) at (-0.75, 0.25) {A};
		\node [style=none] (43) at (7.25, 0) {=};
		\node [style=none] (44) at (1.75, 0) {};
		\node [style=none] (45) at (3, 0.25) {C};
		\node [style=none] (46) at (4.5, 0.5) {};
		\node [style=none] (47) at (4.5, -0.5) {};
		\node [style=none] (48) at (5.25, -0.5) {};
		\node [style=none] (49) at (5.25, 0.5) {};
		\node [style=nodonero] (50) at (4, 0) {};
		\node [style=none] (51) at (5.25, -0.5) {};
		\node [style=none] (52) at (8, 0) {};
		\node [style=none] (53) at (8.5, 0.25) {C};
		\node [style=box] (54) at (9.5, 0) {$h$};
		\node [style=none] (55) at (11.25, 0) {};
		\node [style=box] (56) at (5.25, 0.5) {$h$};
		\node [style=box] (57) at (5.25, -0.5) {$g$};
		\node [style=none] (59) at (6, 0.5) {};
		\node [style=none] (60) at (5.75, 1) {};
		\node [style=none] (61) at (5.75, -1) {};
		\node [style=none] (62) at (5.75, -1) {B};
		\node [style=none] (63) at (5.75, 1) {A};
		\node [style=none] (64) at (10.5, 0.25) {};
		\node [style=none] (65) at (10.5, 0.25) {B};
		\node [style=nodonero] (66) at (6, 0.5) {};
		\node [style=none] (67) at (6.5, -0.5) {};
	\end{pgfonlayer}
	\begin{pgfonlayer}{edgelayer}
		\draw [bend left, looseness=1.25] (21) to (17.center);
		\draw (17.center) to (20.center);
		\draw [bend right] (21) to (18.center);
		\draw (18.center) to (19.center);
		\draw (14.center) to (21);
		\draw (24.center) to (31);
		\draw (31) to (32.center);
		\draw (33) to (36.center);
		\draw (34) to (35);
		\draw [bend left, looseness=1.25] (50) to (46.center);
		\draw (46.center) to (49.center);
		\draw [bend right] (50) to (47.center);
		\draw (47.center) to (48.center);
		\draw (44.center) to (50);
		\draw (52.center) to (54);
		\draw (54) to (55.center);
		\draw (56) to (59.center);
		\draw (57) to (67.center);
	\end{pgfonlayer}
\end{tikzpicture}
\]
Now notice that the first equation holds because, since $g$ is total, we have that
\[
\begin{tikzpicture}[scale=0.60, transform shape]
	\begin{pgfonlayer}{nodelayer}
		\node [style=none] (0) at (-4, 0) {=};
		\node [style=none] (14) at (-9.25, 0) {};
		\node [style=none] (15) at (-8.25, 0.25) {C};
		\node [style=none] (17) at (-6.75, 0.5) {};
		\node [style=none] (18) at (-6.75, -0.5) {};
		\node [style=none] (19) at (-6, -0.5) {};
		\node [style=none] (20) at (-6, 0.5) {};
		\node [style=nodonero] (21) at (-7.25, 0) {};
		\node [style=none] (23) at (-6, -0.5) {};
		\node [style=none] (24) at (3, 0) {};
		\node [style=none] (30) at (3.5, 0.25) {C};
		\node [style=box] (31) at (4.5, 0) {$h$};
		\node [style=none] (32) at (6.25, 0) {};
		\node [style=box] (33) at (-6, 0.5) {$h$};
		\node [style=box] (34) at (-6, -0.5) {$g$};
		\node [style=nodonero] (35) at (-5.25, -0.5) {};
		\node [style=none] (36) at (-4.75, 0.5) {};
		\node [style=none] (37) at (-5.5, 1) {};
		\node [style=none] (38) at (-5.5, -1) {};
		\node [style=none] (39) at (-5.5, -1) {B};
		\node [style=none] (40) at (-5.5, 1) {A};
		\node [style=none] (41) at (5.5, 0.25) {};
		\node [style=none] (42) at (5.5, 0.25) {A};
		\node [style=none] (43) at (-3.25, 0) {};
		\node [style=none] (44) at (-2.25, 0.25) {C};
		\node [style=none] (45) at (-0.75, 0.5) {};
		\node [style=none] (46) at (-0.75, -0.5) {};
		\node [style=none] (47) at (0, -0.5) {};
		\node [style=none] (48) at (0, 0.5) {};
		\node [style=nodonero] (49) at (-1.25, 0) {};
		\node [style=box] (51) at (0, 0.5) {$h$};
		\node [style=nodonero] (53) at (0.75, -0.5) {};
		\node [style=none] (54) at (1.25, 0.5) {};
		\node [style=none] (55) at (0.5, 1) {};
		\node [style=none] (56) at (0.5, -1) {};
		\node [style=none] (57) at (0.5, -1) {C};
		\node [style=none] (58) at (0.5, 1) {A};
		\node [style=none] (59) at (2, 0) {=};
	\end{pgfonlayer}
	\begin{pgfonlayer}{edgelayer}
		\draw [bend left, looseness=1.25] (21) to (17.center);
		\draw (17.center) to (20.center);
		\draw [bend right] (21) to (18.center);
		\draw (18.center) to (19.center);
		\draw (14.center) to (21);
		\draw (24.center) to (31);
		\draw (31) to (32.center);
		\draw (33) to (36.center);
		\draw (34) to (35);
		\draw [bend left, looseness=1.25] (49) to (45.center);
		\draw (45.center) to (48.center);
		\draw [bend right] (49) to (46.center);
		\draw (46.center) to (47.center);
		\draw (43.center) to (49);
		\draw (51) to (54.center);
		\draw (47.center) to (53);
	\end{pgfonlayer}
\end{tikzpicture}
\]
%
Similarly one checks that the second equality holds, and hence conclude the proof.

\item For any two objects $A$ and $B$ of $\mC$, we show that $A\ox B$ is the categorical product of $A$ and $B$ in $\fn{\mC}$ with projections $\freccia{A\ox B}{\id_A\ox\,!_B}{A}$ and $\freccia{A\ox B}{!_A\ox {}\id_B}{B}$. First, by the previous point, we know that this defines a weak product in $\total{\mC}$, and thus in particular in $\fn{\mC}$, since if $f$ and $g$ are $\mC$-functional and $\mC$-total, then again so is $(f \otimes g)\nabla_C$ by Proposition~\ref{prop: ! and nabla are total and functional}.

We still have to prove the uniqueness part of the universal property. So suppose that there exists another $\mC$-total, $\mC$-functional arrow $\freccia{C}{h}{A\ox B}$ such that $(\id_A\ox \,!_B)h=f$ and $(!_A\ox {}\id_B)h=g$. Then we have
\begin{align*}
	(f\ox g)\nabla_C & = \left[(\id_A\ox \,!_B)h\ox (!_A\ox \id_B)h\right]\nabla_C	\\[2pt]
		& = (\id_A\ox \,!_B\,\ox \,!_A\ox {}\id_B)(h\ox h)\nabla_C.
\end{align*}
Since $h$ is $\mC$-functional, we can further evaluate to
\[
 	(f\ox g)\nabla_C=(\id_A\ox \,!_B\, \ox \,!_A\ox {}\id_B)\nabla_{A\ox B}h=h.
	\qedhere
\]
\end{enumerate}

\end{proof}

\comment{
First, on the subcategory of total arrows $\total{\mC}$, we get the structure of weak cartesian 
products. 

\begin{myproposition}\label{thm C gs-cat implies C-total maps have weak cart prod}
Let us consider a gs-monoidal category $\mC$. Then for any $A,B \in \mC$, the monoidal product $A \ox B$ is a weak product in the subcategory $\total{\mC}$.
\end{myproposition}
\begin{proof}
See Appendix \ref{proof: thm C gs-cat implies C-total maps have weak cart prod}.
\end{proof}
}

\begin{myremark}
Let us consider a gs-monoidal category $\mC$ and a $\mC$-total arrow $\freccia{A}{f}{B}$. Notice that the diagram
\[\xymatrix@+1pc{
B & B\ox B \ar[l]_{\id_B\ox \,!_B} \ar[r]^{!_B\ox {}\id_B} & B\\
& A\ar[u]^h\ar[ru]_f\ar[lu]^f
}\]
commutes with both $h:=\nabla_Bf$ and $h:=(f\otimes f)\nabla_A$. If $f$ is non-functional, then these arrows are different. So the subcategory $\total{\mC}$ has weak binary products by Proposition~\ref{thm C gs-cat implies C-total maps have weak cart prod}, but in general these are not categorical products.
\end{myremark}

\comment{
In order to obtain the universal property of products, we have to restrict not only to $\mC$-total arrows as in Theorem~\ref{thm C gs-cat implies C-total maps have weak cart prod}, but also to $\mC$-functional arrows.
\begin{myproposition}\label{thm total functions have products}
Let us consider a gs-monoidal category $\mC$. Then for any $A,B \in \mC$, the monoidal product $A \ox B$ is a categorical product in the subcategory $\fn{\mC}$.
\end{myproposition}
\begin{proof}
See Appendix \ref{proof: thm total functions have products}.
\end{proof}
}

Let us also note that gs-monoidal categories have enough structure to properly express a notion of \emph{domain} of arrows.
\begin{mydefinition}\label{def: domain}
Let $\mC$ be a gs-monoidal category and $\freccia{A}{f}{B}$ in $\mC$. The \textbf{domain} of $f$ is the arrow 
$\freccia{A}{\dom(f):= (\id_A\otimes \,!_Bf)\nabla_A}{A}$, graphically
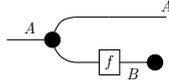
\[
\begin{tikzpicture}[scale=0.60, transform shape]
	\begin{pgfonlayer}{nodelayer}
		\node [style=none] (0) at (-3.25, 1.5) {};
		\node [style=none] (1) at (-1.75, 2) {};
		\node [style=none] (2) at (-1.75, 1) {};
		\node [style=none] (3) at (-1, 1) {};
		\node [style=none] (4) at (-1, 2) {};
		\node [style=nodonero] (5) at (-2.25, 1.5) {};
		\node [style=none] (6) at (-2.75, 1.75) {A};
		\node [style=none] (7) at (-1, 2) {};
		\node [style=box] (14) at (-1, 1) {$f$};
		\node [style=nodonero] (15) at (0, 1) {};
		\node [style=none] (16) at (0.25, 2) {};
		\node [style=none] (19) at (0.25, 2.25) {A};
		\node [style=none] (20) at (-0.5, 0.75) {B};
	\end{pgfonlayer}
	\begin{pgfonlayer}{edgelayer}
		\draw (0.center) to (5);
		\draw [bend left, looseness=1.25] (5) to (1.center);
		\draw (1.center) to (4.center);
		\draw [bend right] (5) to (2.center);
		\draw (2.center) to (3.center);
		\draw (3.center) to (14.center);
		\draw (14) to (15);
		\draw (7.center) to (16.center);
	\end{pgfonlayer}
\end{tikzpicture}
\]
\end{mydefinition}

The motivation behind this particular definition is that in $\Rel$, for a relation $R$ from $A$ to $B$, the arrow $\dom (R)=(\id_A\times \,!_B R)\nabla_A$ is the relation representing the domain of definition of $R$, i.e.
\[
	a\dom (R)a' \quad\iff\quad a=a' \;\mbox{ and }\; \exists b \in B : \; aRb.
\]

As we will show, the arrow $\dom (f)$ in a gs-monoidal category enjoys several algebraic properties that generalise those 
of the usual notion of domain of a relation.
 A first property we expect to hold from a morphism that abstracts the notion of domain of a relation is that the domain of a total arrow has to be the identity.

\begin{myproposition}\label{prop:dom is functionional and dom(f)=id iff f total}
Let $\mC$ be a gs-monoidal category and $\freccia{A}{f}{B}$. Then
\begin{enumerate}
\item if $f$ is $\mC$-functional, then ${\dom (f)}$ is $\mC$-functional;
\item $f$ is $\mC$-total if and only if $\dom (f)=\id_A$.
\end{enumerate}
\end{myproposition}
\begin{proof}
	\begin{enumerate}
	\item By Proposition \ref{prop: ! and nabla are total and functional} applied to the definition of $\dom(f)$.

	\item If $f$ is $\mC$-total, then $!_Bf={}!_A$, hence
\[\begin{tikzpicture}[scale=0.60, transform shape]
	\begin{pgfonlayer}{nodelayer}
		\node [style=none] (0) at (-3.5, 1.5) {};
		\node [style=none] (1) at (-2, 2) {};
		\node [style=none] (2) at (-2, 1) {};
		\node [style=none] (3) at (-1.25, 1) {};
		\node [style=nodonero] (5) at (-2.5, 1.5) {};
		\node [style=none] (6) at (-3, 1.75) {A};
		\node [style=box] (14) at (-1.25, 1) {$f$};
		\node [style=nodonero] (15) at (-0.25, 1) {};
		\node [style=none] (16) at (0, 2) {};
		\node [style=none] (19) at (0, 2.25) {A};
		\node [style=none] (20) at (-0.75, 0.75) {B};
		\node [style=none] (25) at (1, 1.5) {=};
		\node [style=none] (30) at (2, 1.5) {};
		\node [style=none] (31) at (3.5, 2) {};
		\node [style=none] (32) at (3.5, 1) {};
		\node [style=none] (33) at (4.25, 1) {};
		\node [style=nodonero] (34) at (3, 1.5) {};
		\node [style=none] (35) at (2.5, 1.75) {A};
		\node [style=nodonero] (37) at (5, 1) {};
		\node [style=none] (38) at (5.5, 2) {};
		\node [style=none] (39) at (5, 2.25) {A};
		\node [style=none] (40) at (4.5, 0.75) {A};
		\node [style=none] (41) at (7.5, 1.5) {};
		\node [style=none] (42) at (8.25, 1.75) {A};
		\node [style=none] (43) at (6.5, 1.5) {=};
		\node [style=none] (44) at (9, 1.5) {};
	\end{pgfonlayer}
	\begin{pgfonlayer}{edgelayer}
		\draw (0.center) to (5);
		\draw [bend left, looseness=1.25] (5) to (1.center);
		\draw [bend right] (5) to (2.center);
		\draw (2.center) to (3.center);
		\draw (14) to (15);
		\draw (1.center) to (16.center);
		\draw (30.center) to (34);
		\draw [bend left, looseness=1.25] (34) to (31.center);
		\draw [bend right] (34) to (32.center);
		\draw (32.center) to (33.center);
		\draw (31.center) to (38.center);
		\draw (41.center) to (44.center);
		\draw (33.center) to (37);
	\end{pgfonlayer}
\end{tikzpicture}
\]
by the second axiom of gs-monoidal categories. Conversely, if $\dom (f)=\id_A$, then $!_A\dom(f)= {}!_A$, hence $!_A(\id_A\otimes \,!_Bf)\nabla_A={}!_A$. 
Since
\[\begin{tikzpicture}[scale=0.60, transform shape]
	\begin{pgfonlayer}{nodelayer}
		\node [style=none] (0) at (-3.5, 1.5) {};
		\node [style=none] (1) at (-2, 2) {};
		\node [style=none] (2) at (-2, 1) {};
		\node [style=none] (3) at (-1.25, 1) {};
		\node [style=nodonero] (5) at (-2.5, 1.5) {};
		\node [style=none] (6) at (-3, 1.75) {A};
		\node [style=box] (14) at (-1.25, 1) {$f$};
		\node [style=nodonero] (15) at (0, 1) {};
		\node [style=none] (16) at (0, 2) {};
		\node [style=none] (19) at (-0.5, 2.25) {A};
		\node [style=none] (20) at (-0.5, 0.75) {B};
		\node [style=none] (25) at (1, 1.5) {=};
		\node [style=none] (26) at (2, 1.5) {};
		\node [style=none] (27) at (3.25, 1.5) {};
		\node [style=box] (28) at (3.25, 1.5) {$f$};
		\node [style=nodonero] (29) at (4.5, 1.5) {};
		\node [style=none] (30) at (4, 1.25) {B};
		\node [style=none] (31) at (2.5, 1.25) {A};
		\node [style=nodonero] (32) at (0, 2) {};
	\end{pgfonlayer}
	\begin{pgfonlayer}{edgelayer}
		\draw (0.center) to (5);
		\draw [bend left, looseness=1.25] (5) to (1.center);
		\draw [bend right] (5) to (2.center);
		\draw (2.center) to (3.center);
		\draw (14) to (15);
		\draw (1.center) to (16.center);
		\draw (26.center) to (27.center);
		\draw (28) to (29);
	\end{pgfonlayer}
\end{tikzpicture}
\]
we can conclude that $!_Bf={}!_A$, i.e. that $f$ is $\mC$-total.
\qedhere
\end{enumerate}
\end{proof}

\comment{
By the naturality of braiding, checking the following property is immediate.
\begin{mylemma}\label{prop: dom(g x f)=dom(g)xdom (f)}
	Let $\mC$ be a gs-monoidal category. If $\freccia{A}{f}{B}$ and $\freccia{C}{g}{D}$ are $\mC$-functional, then 
	$\dom (f \otimes g) = dom (f) \otimes \dom(g)$.

\end{mylemma}
}
\begin{myremark}
	The notion of domain allows us to state the precise relation between gs-monoidal categories, restriction categories~\cite{Cockett02} and p-categories~\cite{Robinson88}. 
	Given the remarks on commutative comonoids provided in Examples~\ref{comon_is_gs} and~\ref{comon2}, we can exploit a relevant result of Cockett and Lack~\cite[Thm~5.2]{Cockett07} to conclude that when $\mC$ is a gs-monoidal category, the subcategory $\pfn{\mC}$ of $\mC$-functional arrows is a restriction category, with the restriction structure given by $\dom(-)$. Moreover, this category has restriction products, which in particular implies the equation $\dom(f \otimes g) = \dom(f) \otimes \dom(g)$. Furthermore, $\pfn{\mC}$ also is a p-category with one-element object (given by $\I$), where the diagonal is given by $\nabla$ and the two projections by the families of arrows of the form $\id\ox\,!$ and $!\ox\id$, respectively.
	
This fact, combined with the remark in~\cite[p. 29]{Cockett07} that explains how in a restriction category with restriction products every object has a canonical
cocommutative comonoid structure (in the symmetric monoidal category), the duplicator is natural and total maps are precisely the counit-preserving ones (this fact is due to an observation of Carboni~\cite{Carboni_87}), allows to conclude that restriction categories with restriction products correspond exactly to gs-monoidal categories whose duplicator is natural.
\end{myremark}

\begin{myremark}\label{rem fdom(f)=f in rel but not in any gs}
 In $\Rel$, domains have an additional property given by the equation $f\dom (f)=f$, namely
\begin{equation}
\begin{split}
\begin{tikzpicture}[scale=0.60, transform shape]
	\begin{pgfonlayer}{nodelayer}
		\node [style=none] (31) at (4.75, 1.5) {};
		\node [style=none] (32) at (6.25, 2) {};
		\node [style=none] (33) at (6.25, 1) {};
		\node [style=none] (34) at (7, 1) {};
		\node [style=nodonero] (35) at (5.75, 1.5) {};
		\node [style=none] (36) at (5.25, 1.75) {A};
		\node [style=box] (37) at (7, 1) {$f$};
		\node [style=nodonero] (38) at (8, 1) {};
		\node [style=none] (39) at (8.25, 2) {};
		\node [style=none] (40) at (8.25, 2.25) {B};
		\node [style=none] (41) at (7.5, 0.75) {B};
		\node [style=none] (42) at (9.75, 1.5) {};
		\node [style=none] (43) at (10.25, 1.75) {A};
		\node [style=none] (44) at (9, 1.5) {$=$};
		\node [style=box] (45) at (7, 2) {$f$};
		\node [style=box] (46) at (10.75, 1.5) {$f$};
		\node [style=none] (47) at (11.75, 1.5) {};
		\node [style=none] (48) at (11.25, 1.75) {B};
	\end{pgfonlayer}
	\begin{pgfonlayer}{edgelayer}
		\draw (31.center) to (35);
		\draw [bend left, looseness=1.25] (35) to (32.center);
		\draw [bend right] (35) to (33.center);
		\draw (33.center) to (34.center);
		\draw (34.center) to (37.center);
		\draw (37) to (38);
		\draw (32.center) to (45);
		\draw (45) to (39.center);
		\draw (42.center) to (46);
		\draw (46) to (47.center);
	\end{pgfonlayer}
\end{tikzpicture}
\label{eq:fdomf}
\end{split}
\end{equation}
which holds even for non-functional $f$.
However, this equation does not hold for arrows in gs-monoidal categories in general.
For example, the Kleisli category of the multiset monad on $\FinSet$ is a gs-monoidal category in a canonical way (by Proposition~\ref{prop: Kleisli su cartesiane sono gs}), and its arrows $A \to B$ can be identified with functions $A \times B \to \mathbb{N}$ that compose via convolution.
The above~\eqref{eq:fdomf} does not hold there: already with $A = B = I$ a one-element set, where $f$ is determined by a natural number, we have the number itself on the right but its square on the left.
We will return to this issue in Section~\ref{sec:oplax_cart} after introducing preorder-enriched gs-monoidal categories.
\end{myremark}

\comment{

\begin{myproposition}
	\label{prop:gs_vs_restriction}
	Let $\mC$ be a gs-monoidal category. Then the subcategory $\pfn{\mC}$ of $\mC$-functional arrows is a monoidal restriction category, with the restriction structure given by $\dom (-)$.
\end{myproposition}
\tob{Hadn't we concluded that this is already covered by result on restriction products in~\cite{Cockett07}, in particular the equivalence of (i) and (iv) in their Theorem 5.2? The equation $\dom(f \otimes g) = \dom(f) \otimes \dom(g)$, or for them $\overline{f \times g} = \overline{f} \times \overline{g}$, is encoded in their definition of restriction product, which requires in particular $\times$ to be a restriction functor. Our form of $\dom$ appears in the context of p-categories as dom in their Section 4.2. In other words, I think that the description as a restriction category with restriction products is actually stronger than saying it's a monoidal restriction category, and we don't need the latter concept at all}
\fab{proof in the appendix}
The result can be seen as a generalisation of~\cite{xxx}, in order to account for the notion of monoidal restriction category in the sense of~\cite{Heunen2021}, namely, a restriction and monoidal category such that restriction structure satisfies the equation $\dom(f\otimes g)=\dom(f)\otimes \dom(g)$
shown in Lemma~\ref{prop: dom(g x f)=dom(g)xdom (f)}.

The result above allows to establish a link between gs-monoidal categories and \emph{p-categories with one-element object}, as done for
restriction categories (see \cite[Thm.~5.2]{Cockett07} for details). For an introduction to p-categories we refer to \cite{Robinson88}, while 
its definition is recalled in Appendix~\ref{xxx}.

\tob{Given that apparently neither Prop \ref{prop:gs_vs_restriction} nor Prop \ref{thm C-functionals form p-cat} is new, and we don't use p-categories or restriction categories otherwise, how about replacing both of these propositions by a single remark with reference to~\cite[Thm~5.2]{Cockett07}? Then we don't need an additional appendix with their definitions}

\begin{myproposition}\label{thm C-functionals form p-cat}
Let $\mC$ be a gs-monoidal category. 
Then the subcategory $\pfn{\mC}$ of $\mC$-functional arrows is a p-category with one-element object, 
with the one-element object given by $\I$, the diagonal given by $\nabla$, and the two projections given by the 
families of arrows of the form $\id\ox\,!$ and $!\ox\id$, respectively.
\end{myproposition}

If we consider arrows other than functional ones, then the above does not provide the structure of p-category, since $\nabla$ then fails to be natural.

\comment{
The rest of this section is devoted to analysing what structure the tensor product of a gs-monoidal category induces in the three subcategories $\pfn{\mC}$, $\total{\mC}$ and $\fn{\mC}$. We start by considering the subcategory of $\mC$-functional arrows $\pfn{\mC}$.

The notion of \emph{p-category} was introduced by Rosolini and Robinson in \cite{Robinson88} to present categories of partial maps. Such a notion can be considered an improvement on the \emph{dominical categories} by Di Paola and Heller \cite{DiPaola87}.
A p-category is a category $\mC$ endowed with a bifunctor $\freccia{\mC\times \mC}{\times}{\mC}$, called \emph{product}, 
a natural transformation $\freccia{\id_{\mC}}{\nabla}{\id_{\mC}\times \id_{\mC}} $, called \emph{diagonal}, and two families of natural 
transformations
\[
	p_{(-),Y}: (-)\times Y\to (-), \qquad q_{X,(-)}: X\times (-)\to (-)
\]
indexed by objects $X, Y \in \mC$ and called \emph{projections}, satisfying some axioms (see \cite{Robinson88}).
In a p-category $\mC$, a \emph{one-element object} (see \cite[Sec.~1]{Robinson88}) is an object $T$ of $\mC$ together with a family $t_X\to T$ of arrows of $\mC$ for which each $p_{X,T}$ is invertible, and whose inverse is $\id_X\times t_X \nabla$. 

In the following theorem we recall, using the language and the notation of gs-monoidal categories, a known result that establishes the link between gs-monoidal categories in which every arrow is functional and p-categories with one-element object, and we refer to \cite[Thm.~5.2]{Cockett07} for all the details. 

 
\begin{mytheorem}\label{thm C-functionals form p-cat}
Let us consider a gs-monoidal category $\mC$. Then the monoidal operation $\ox$, the natural transformation $\nabla$ and the projections given by arrows of the form $\id\ox\,!$ and $!\ox\id$ equip $\pfn{\mC}$ with the structure of p-category with one-element object.
\end{mytheorem}

If we consider arrows other than functional ones, then the above does not provide the structure of p-category, since $\nabla$ then fails to be natural.
}
}

\section{Oplax cartesian categories}
\label{sec:oplax_cart}


We have seen that gs-monoidal categories enjoy some features of $\Rel$ with respect to total and functional arrows. 
Our next step is to show how to build on the notion of gs-monoidality in order to
account for the usual preorder-enrichment of $\Rel$.
This extension will be pivotal later on, e.g. for our functorial completeness theorem.
And while the notion of domain we discussed in the previous section lacks some of the properties of domains in $\Rel$,
in particular the one noted in Remark \ref{rem fdom(f)=f in rel but not in any gs}, 
we will see that this property can be recovered in terms of preorder enrichment.


\begin{mydefinition}\label{def preorder-enriched gs-monoidal category}
	A \textbf{preorder-enriched gs-monoidal category} $\mC$ is a gs-monoidal category $\mC$ that is at the same time a preorder-enriched monoidal category.
\end{mydefinition}

Recall that a preorder-enriched monoidal category consists of a preorder-enriched category $\mC$, an object $I$ of $\mC$, a preorder-enriched functor $\freccia{\mC\times \mC}{\otimes}{\mC}$,  and enriched natural monoidal structure
isomorphisms
\[
	\freccia{I\otimes -}{\lambda}{\id_{\mC}}, \qquad \freccia{- {}\otimes{} I}{\rho}{\id_{\mC}}, \qquad \freccia{(-\otimes -)\otimes -}{\alpha}{-\otimes (-\otimes -)}
\]
such that the underlying category equipped with the underlying functor $\otimes$, the object $I$, and the natural isomorphisms $\lambda,\rho$ and $\alpha$ is a monoidal category (see \cite{Kelly05} for details).
Since a preorder-enriched functor is just an ordinary functor that is in addition monotone, the preorder structure and the monoidal structure are required to interact by the monotonicity of the tensor product $\otimes$; the preorder-enrichment of the structure isomorphisms $\lambda$, $\rho$ and $\alpha$ does not add any additional conditions since preorder-enrichment for natural transformations between preorder-enriched functors is trivial.

In a general preorder-enriched gs-monoidal category, no further compatibility with the gs-monoidal structure is required.
However, in many (but not all) examples, also the following compatibility holds.

\begin{mydefinition}\label{def oplax cartesian cat}
An \textbf{oplax cartesian category} $\mC$ is a preorder-enriched gs-monoidal category $\mC$ such that the following inequalities 
hold\footnote{Viewing a preorder-enriched category as a $2$-category, the inequalities state that the families of arrows $\nabla_A$ 
and $!_A$ are the components of an \emph{oplax natural transformation}. The connection between these inequalities and rewriting is explored in~\cite{CorradiniGadducci99b}.}
for every arrow $\freccia{A}{f}{B}$
\begin{figure}[H]
\centering
\begin{tikzpicture}[scale=0.60, transform shape]
    \begin{pgfonlayer}{nodelayer}
		\node [style=none] (0) at (-5, 0) {$\leq$};
		\node [style=none] (1) at (-4, -2) {};
		\node [style=nodonero] (2) at (-3, -2) {};
		\node [style=none] (3) at (-3.5, -1.75) {A};
		\node [style=box] (5) at (-7.5, -2) {$f$};
		\node [style=nodonero] (7) at (-6.25, -2) {};
		\node [style=none] (8) at (-6.75, -1.75) {B};
		\node [style=none] (9) at (-9, -2) {};
		\node [style=none] (10) at (-8.25, -1.75) {A};
		\node [style=box] (11) at (-8.5, 0) {$f$};
		\node [style=none] (13) at (-7.75, 0.25) {B};
		\node [style=none] (14) at (-9.75, 0) {};
		\node [style=none] (15) at (-9.25, 0.25) {A};
		\node [style=none] (17) at (-6.75, 0.5) {};
		\node [style=none] (18) at (-6.75, -0.5) {};
		\node [style=none] (19) at (-6, -0.5) {};
		\node [style=none] (20) at (-6, 0.5) {};
		\node [style=nodonero] (21) at (-7.25, 0) {};
		\node [style=none] (23) at (-6, -0.5) {};
		\node [style=none] (24) at (-4.25, 0) {};
		\node [style=none] (25) at (-2.75, 0.5) {};
		\node [style=none] (26) at (-2.75, -0.5) {};
		\node [style=nodonero] (29) at (-3.25, 0) {};
		\node [style=none] (30) at (-3.75, 0.25) {A};
		\node [style=box] (32) at (-1.5, -0.5) {$f$};
		\node [style=none] (34) at (-0.75, -0.75) {B};
		\node [style=none] (35) at (-2.75, -0.5) {};
		\node [style=none] (36) at (-2.25, -0.75) {};
		\node [style=none] (37) at (-0.25, -0.5) {};
		\node [style=box] (39) at (-1.5, 0.5) {$f$};
		\node [style=none] (40) at (-0.75, 0.75) {B};
		\node [style=none] (42) at (-2.25, 0.75) {};
		\node [style=none] (43) at (-0.25, 0.5) {};
		\node [style=none] (44) at (-5, -2) {$\leq$};
	\end{pgfonlayer}
	\begin{pgfonlayer}{edgelayer}
		\draw (1.center) to (2);
		\draw (5) to (7);
		\draw (9.center) to (5);
		\draw (14.center) to (11);
		\draw [bend left, looseness=1.25] (21) to (17.center);
		\draw (17.center) to (20.center);
		\draw [bend right] (21) to (18.center);
		\draw (18.center) to (19.center);
		\draw (11) to (21);
		\draw (24.center) to (29);
		\draw [bend left, looseness=1.25] (29) to (25.center);
		\draw [bend right] (29) to (26.center);
		\draw (35.center) to (32);
		\draw (32) to (37.center);
		\draw (39) to (43.center);
		\draw (25.center) to (39);
	\end{pgfonlayer}
\end{tikzpicture}
\end{figure}
\end{mydefinition}
The notion of oplax cartesian category is reminiscent of cartesian bicategories in the sense of~\cite[Definition~1.2]{CARBONI198711}. 
In particular, there are two differences between these two notions:  the first one is that the existence of right adjoints for $\nabla$ and $!$ is not required in Definition \ref{def oplax cartesian cat}, while it is a crucial part of the definition of a cartesian bicategory (see point $(M)$ of ~\cite[Definition~1.2]{CARBONI198711}). The second one is that the hom-categories of an oplax cartesian category are required to be preorders, while cartesian bicategories are originally required to be poset-enriched.

\begin{myexample}
$\Rel$ has a natural preorder-enriched structure with the preorder given by the set-theoretic inclusions between relations. Moreover, for every relation $\freccia{A}{R}{B}$ we trivially have the inequalities discussed in Definition~\ref{def oplax cartesian cat}. Hence $\Rel$ is an oplax cartesian category.
On the other hand, reversing the preorder on every hom-set of $\Rel$ gives a preorder-enriched gs-monoidal category that is clearly not oplax cartesian.

For a more trivial example, any gs-monoidal category $\mC$ is preorder-enriched with the trivial preorder. Then $\mC$ is oplax cartesian if and only if it is cartesian monoidal.
\end{myexample}

\begin{mydefinition}\label{def_preorder_equivalence}
Let $\mC$ be a preorder-enriched category and $\freccia{A}{f,g}{B}$. Then $f$ and $g$ are \textbf{preorder equivalent}, denoted by $f\approx g$,
if $f\leq g$ and $g\leq f$. 
\end{mydefinition}
Definition~\ref{def: C-fun and C-tot} can be generalised straightforwardly.
\begin{mydefinition}\label{def_weakly_total_and_function}
	Let $\mC$ be an oplax cartesian category. An arrow $\stringdiagfreccia{A}{f}{B}$ is
	\textbf{weakly} $\mC$-\textbf{total} if
	\[
	\begin{tikzpicture}[scale=0.60, transform shape]
		\begin{pgfonlayer}{nodelayer}
			\node [style=none] (0) at (3.75, 0) {$\approx$};
			\node [style=none] (1) at (4.5, 0) {};
			\node [style=nodonero] (2) at (5.5, 0) {};
			\node [style=none] (3) at (5, 0.25) {A};
			\node [style=box] (5) at (1.75, 0) {$f$};
			\node [style=nodonero] (7) at (3, 0) {};
			\node [style=none] (8) at (2.5, 0.25) {B};
			\node [style=none] (9) at (0.5, 0) {};
			\node [style=none] (10) at (1, 0.25) {A};
		\end{pgfonlayer}
		\begin{pgfonlayer}{edgelayer}
			\draw (1.center) to (2);
			\draw (5) to (7);
			\draw (9.center) to (5);
		\end{pgfonlayer}
	\end{tikzpicture}
	\]
	and \textbf{weakly} $\mC$-\textbf{functional} if
	\[
	\begin{tikzpicture}[scale=0.60, transform shape]
		\begin{pgfonlayer}{nodelayer}
			\node [style=none] (0) at (-5.75, 0) {$\approx$};
			\node [style=box] (11) at (-8.75, 0) {$f$};
			\node [style=none] (13) at (-8, 0.25) {B};
			\node [style=none] (14) at (-10, 0) {};
			\node [style=none] (15) at (-9.5, 0.25) {A};
			\node [style=none] (17) at (-7, 0.5) {};
			\node [style=none] (18) at (-7, -0.5) {};
			\node [style=none] (19) at (-6.25, -0.5) {};
			\node [style=none] (20) at (-6.25, 0.5) {};
			\node [style=nodonero] (21) at (-7.5, 0) {};
			\node [style=none] (23) at (-6.25, -0.5) {};
			\node [style=none] (24) at (-5.25, 0) {};
			\node [style=none] (25) at (-3.75, 0.5) {};
			\node [style=none] (26) at (-3.75, -0.5) {};
			\node [style=nodonero] (29) at (-4.25, 0) {};
			\node [style=none] (30) at (-4.75, 0.25) {A};
			\node [style=box] (32) at (-2.5, -0.5) {$f$};
			\node [style=none] (34) at (-1.75, -0.75) {B};
			\node [style=none] (35) at (-3.75, -0.5) {};
			\node [style=none] (36) at (-3.25, -0.75) {};
			\node [style=none] (37) at (-1.25, -0.5) {};
			\node [style=box] (39) at (-2.5, 0.5) {$f$};
			\node [style=none] (40) at (-1.75, 0.75) {B};
			\node [style=none] (42) at (-3.25, 0.75) {};
			\node [style=none] (43) at (-1.25, 0.5) {};
		\end{pgfonlayer}
		\begin{pgfonlayer}{edgelayer}
			\draw (14.center) to (11);
			\draw [bend left, looseness=1.25] (21) to (17.center);
			\draw (17.center) to (20.center);
			\draw [bend right] (21) to (18.center);
			\draw (18.center) to (19.center);
			\draw (11) to (21);
			\draw (24.center) to (29);
			\draw [bend left, looseness=1.25] (29) to (25.center);
			\draw [bend right] (29) to (26.center);
			\draw (35.center) to (32);
			\draw (32) to (37.center);
			\draw (39) to (43.center);
			\draw (25.center) to (39);
		\end{pgfonlayer}
	\end{tikzpicture}
	\]
	\end{mydefinition}
Definition~\ref{def_preorder_equivalence} lets us prove that, even if $f\dom (f)$ is different from $f$ in general, these two arrows are preorder equivalent in any  oplax cartesian category.

\begin{myproposition}\label{prop: fdom(f)=gs f}
Let $\mC$ be an oplax cartesian category. For every arrow $\freccia{A}{f}{B}$ we have that $\dom(f) \leq \id_A$ and $f\dom (f)\approx f$, graphically
\begin{figure}[H]
\centering
\begin{tikzpicture}[scale=0.60, transform shape]

	\begin{pgfonlayer}{nodelayer}
		\node [style=none] (0) at (-3.5, 1.5) {};
		\node [style=none] (1) at (-2, 2) {};
		\node [style=none] (2) at (-2, 1) {};
		\node [style=none] (3) at (-1.25, 1) {};
		\node [style=nodonero] (5) at (-2.5, 1.5) {};
		\node [style=none] (6) at (-3, 1.75) {A};
		\node [style=box] (14) at (-1.25, 1) {$f$};
		\node [style=nodonero] (15) at (-0.25, 1) {};
		\node [style=none] (16) at (0, 2) {};
		\node [style=none] (19) at (0, 2.25) {A};
		\node [style=none] (20) at (-0.75, 0.75) {B};
		\node [style=none] (21) at (2, 1.5) {};
		\node [style=none] (23) at (2.75, 1.75) {A};
		\node [style=none] (25) at (1, 1.5) {$\leq$};
		\node [style=none] (29) at (3.5, 1.5) {};
		\node [style=none] (31) at (-3.5, -1) {};
		\node [style=none] (32) at (-2, -0.5) {};
		\node [style=none] (33) at (-2, -1.5) {};
		\node [style=none] (34) at (-1.25, -1.5) {};
		\node [style=nodonero] (35) at (-2.5, -1) {};
		\node [style=none] (36) at (-3, -0.75) {A};
		\node [style=box] (37) at (-1.25, -1.5) {$f$};
		\node [style=nodonero] (38) at (-0.25, -1.5) {};
		\node [style=none] (39) at (0, -0.5) {};
		\node [style=none] (40) at (0, -0.25) {B};
		\node [style=none] (41) at (-0.75, -1.75) {B};
		\node [style=none] (42) at (2, -1) {};
		\node [style=none] (43) at (2.5, -0.75) {A};
		\node [style=none] (44) at (1, -1) {$\approx $};
		\node [style=box] (45) at (-1.25, -0.5) {$f$};
		\node [style=box] (46) at (3, -1) {$f$};
		\node [style=none] (47) at (4, -1) {};
		\node [style=none] (48) at (3.5, -0.75) {B};
	\end{pgfonlayer}
	\begin{pgfonlayer}{edgelayer}
		\draw (0.center) to (5);
		\draw [bend left, looseness=1.25] (5) to (1.center);
		\draw [bend right] (5) to (2.center);
		\draw (2.center) to (3.center);
		\draw (14) to (15);
		\draw (31.center) to (35);
		\draw [bend left, looseness=1.25] (35) to (32.center);
		\draw [bend right] (35) to (33.center);
		\draw (33.center) to (34.center);
		\draw (37) to (38);
		\draw (32.center) to (45);
		\draw (45) to (39.center);
		\draw (42.center) to (46);
		\draw (46) to (47.center);
		\draw (1.center) to (16.center);
		\draw (21.center) to (29.center);
	\end{pgfonlayer}
\end{tikzpicture}

\end{figure}
\end{myproposition}
\begin{proof}
\begin{enumerate}
	\item By definition of oplax cartesian category we have the inequality $\duefreccianoname{!_Bf}{\,!_A}$, and thus
	
	\[\begin{tikzpicture}[scale=0.60, transform shape]
	\begin{pgfonlayer}{nodelayer}
		\node [style=none] (0) at (-3.5, 1.5) {};
		\node [style=none] (1) at (-2, 2) {};
		\node [style=none] (2) at (-2, 1) {};
		\node [style=none] (3) at (-1.25, 1) {};
		\node [style=nodonero] (5) at (-2.5, 1.5) {};
		\node [style=none] (6) at (-3, 1.75) {A};
		\node [style=box] (14) at (-1.25, 1) {$f$};
		\node [style=nodonero] (15) at (-0.25, 1) {};
		\node [style=none] (16) at (0, 2) {};
		\node [style=none] (19) at (0, 2.25) {A};
		\node [style=none] (20) at (-0.75, 0.75) {B};
		\node [style=none] (21) at (7, 1.5) {};
		\node [style=none] (23) at (7.75, 1.75) {A};
		\node [style=none] (25) at (1, 1.5) {$\leq$};
		\node [style=none] (29) at (8.5, 1.5) {};
		\node [style=none] (30) at (2, 1.5) {};
		\node [style=none] (31) at (3.5, 2) {};
		\node [style=none] (32) at (3.5, 1) {};
		\node [style=nodonero] (34) at (3, 1.5) {};
		\node [style=none] (35) at (2.5, 1.75) {A};
		\node [style=nodonero] (37) at (4.5, 1) {};
		\node [style=none] (38) at (5, 2) {};
		\node [style=none] (39) at (4.25, 2.25) {A};
		\node [style=none] (40) at (4, 0.75) {A};
		\node [style=none] (41) at (6, 1.5) {=};
	\end{pgfonlayer}
	\begin{pgfonlayer}{edgelayer}
		\draw (0.center) to (5);
		\draw [bend left, looseness=1.25] (5) to (1.center);
		\draw [bend right] (5) to (2.center);
		\draw (2.center) to (3.center);
		\draw (14) to (15);
		\draw (1.center) to (16.center);
		\draw (21.center) to (29.center);
		\draw (30.center) to (34);
		\draw [bend left, looseness=1.25] (34) to (31.center);
		\draw [bend right] (34) to (32.center);
		\draw (31.center) to (38.center);
		\draw (32.center) to (37);
	\end{pgfonlayer}
\end{tikzpicture}
	\]
	
%
	\item The inequality $\duefreccianoname{f\dom (f) }{f}$ follows by the previous point.
We can prove the reverse $\duefreccianoname{f}{f\dom (f)}$ by observing that

\[\begin{tikzpicture}[scale=0.60, transform shape]
	\begin{pgfonlayer}{nodelayer}
		\node [style=none] (26) at (2, 1.5) {};
		\node [style=none] (27) at (3.25, 1.5) {};
		\node [style=box] (28) at (3.25, 1.5) {$f$};
		\node [style=none] (30) at (4, 1.25) {B};
		\node [style=none] (31) at (2.5, 1.25) {A};
		\node [style=none] (33) at (4.5, 1.5) {};
		\node [style=none] (34) at (5.5, 1.5) {=};
		\node [style=none] (35) at (6.5, 1.5) {};
		\node [style=none] (36) at (7.75, 1.5) {};
		\node [style=box] (37) at (7.75, 1.5) {$f$};
		\node [style=none] (38) at (8.5, 1.25) {B};
		\node [style=none] (39) at (7, 1.25) {A};
		\node [style=none] (40) at (9, 1.5) {};
		\node [style=none] (42) at (9.5, 2) {};
		\node [style=none] (43) at (9.5, 1) {};
		\node [style=none] (44) at (10, 1) {};
		\node [style=nodonero] (45) at (9, 1.5) {};
		\node [style=nodonero] (48) at (10, 1) {};
		\node [style=none] (49) at (10.25, 2) {};
		\node [style=none] (50) at (12.25, 1.5) {};
		\node [style=none] (51) at (13.75, 2) {};
		\node [style=none] (52) at (13.75, 1) {};
		\node [style=none] (53) at (14.5, 1) {};
		\node [style=nodonero] (54) at (13.25, 1.5) {};
		\node [style=none] (55) at (12.75, 1.75) {A};
		\node [style=box] (56) at (14.5, 1) {$f$};
		\node [style=nodonero] (57) at (15.5, 1) {};
		\node [style=none] (58) at (15.75, 2) {};
		\node [style=none] (59) at (15.25, 2.25) {B};
		\node [style=none] (60) at (15, 0.75) {B};
		\node [style=box] (62) at (14.5, 2) {$f$};
		\node [style=none] (64) at (11.25, 1.5) {$\leq$};
	\end{pgfonlayer}
	\begin{pgfonlayer}{edgelayer}
		\draw (26.center) to (27.center);
		\draw (28) to (33.center);
		\draw (35.center) to (36.center);
		\draw (37) to (40.center);
		\draw [bend left, looseness=1.25] (45) to (42.center);
		\draw [bend right] (45) to (43.center);
		\draw (43.center) to (44.center);
		\draw (42.center) to (49.center);
		\draw (50.center) to (54);
		\draw [bend left, looseness=1.25] (54) to (51.center);
		\draw [bend right] (54) to (52.center);
		\draw (52.center) to (53.center);
		\draw (56) to (57);
		\draw (51.center) to (58.center);
	\end{pgfonlayer}
\end{tikzpicture}
\]
The first equality holds by the last axiom in \eqref{comm_comon}, while the inequality follows by definition of oplax cartesian category. 
\qedhere
\end{enumerate}
\end{proof}

We next provide a uniqueness result for duplicators and dischargers.

\begin{myproposition}\label{prop:unicity oplax cartesian cat}
	Let $\mC$ be an oplax cartesian category with structure arrows $\nabla$ and $!$. If $\mC$ admits the structure of oplax cartesian category given by the same monoidal structure and the same preorder, but with structure arrows $\nabla'$ and $!'$, then $\nabla_A\approx \nabla_A'$ and $!_A\approx  {}!_A'$ for every object $A$ of $\mC$.
\end{myproposition}
\begin{proof}
	Let us consider the two operators $!$ and $!'$.
	For every object $A$, we have
	\[
		!'_A = {} !_I !'_A \leq {}!_A,
	\]
	where the first equation is simply by $!_I = \id_I$ and the inequality is by oplax cartesianity.
We then obtain $!'_A \approx  {}!_A$ by symmetry.

With a similar argument we can prove that $\nabla_A\approx \nabla'_A$.
\end{proof}

\begin{myexample}
	\label{lax_cartesian_generated}
	Given a (plain) gs-monoidal category $\mC$, one can consider the monoidal preorder enrichment \emph{generated} by the inequalities of Definition~\ref{def oplax cartesian cat}, i.e.~the smallest preorder on every hom-set which makes both composition and $\otimes$ monotone and satisfies those inequalities.
	In this way, $\mC$ becomes oplax cartesian in a canonical way.
	
	An interesting case of this construction
	is $\mathbf{FinStoch}$, the Markov category of finite sets and stochastic maps. 
	In this case, it was shown by Dario Stein (personal communication) that the preorder enrichment generated like this recovers the \emph{support} of a stochastic matrix: 
	for stochastic matrices $f, g : X \to Y$ we have $f \le g$ if and only if $f(y|x) > 0$ implies $g(y|x) > 0$ for all $x \in X$ and $y \in Y$.
	It thus follows that $\mathbf{FinStoch} / \!\!\!\approx$, the quotient of $\mathbf{FinStoch}$ by preorder equivalence, is a gs-monoidal category isomorphic to $\FinRel$.
\end{myexample}

We now conclude this section with a discussion of functors between preorder-enriched gs-monoidal categories and oplax cartesian categories.
The various notions of gs-monoidal functor introduced in Definition \ref{def gs monoidal functor} can be also used in the context of preorder-enriched gs-monoidal categories, with the only difference that $F$ is additionally required to be a preorder-enriched functor, which amounts to monotonicity on hom-sets.
On functors between oplax cartesian categories, one often has additional inequalities, which take the following form.

	\begin{mydefinition}\label{def:(op) oplax cartesian functor}
		Let $\mC$ and $\mD$
		be oplax cartesian categories and  
		$\freccia{\mC}{F}{\mD}$ a preorder-enriched functor. Then
		\begin{enumerate}
			\item $F$ is \textbf{colax cartesian}\footnote{The use of ``colax'' here refers to the direction of the 2-cell, namely from $F(\nabla_A)$ to $\psi_{A,A}\circ \nabla_{FA}$. } if it is a lax symmetric monoidal functor with structure arrows $\psi, \psi_0$ 
			such that the following inequalities hold 
                \[\begin{tikzcd}[column sep=tiny]
                    {F(A)} && {F(A\otimes A)}   &&& FA && {F(I)} \\
                    & {F(A)\otimes F(A)} &&&&& I
                    \arrow[""{name=0, anchor=center, inner sep=0}, "{F(\nabla_A)}", from=1-1, to=1-3]
                    \arrow["{\nabla_{FA}}"', from=1-1, to=2-2]
                    \arrow["{\psi_{A,A}}"', from=2-2, to=1-3]
                    \arrow[""{name=1, anchor=center, inner sep=0}, "{F(!_A)}", from=1-6, to=1-8]
                    \arrow["{!_{FA}}"', from=1-6, to=2-7]
                   \arrow["{\psi_0}"', from=2-7, to=1-8]
                         \arrow["\leq"{marking}, draw=none, from=0, to=2-2]
		        \arrow["\leq"{marking, xshift=-2pt}, draw=none, from=1, to=2-7]
                \end{tikzcd}\] 
			\item $F$ is \textbf{colax opcartesian} if it is an oplax symmetric monoidal functor with structure arrows $\phi, \phi_0$ 
			such that the following inequalities hold 
                \[\begin{tikzcd}[column sep=tiny]
                 {F(A)} && {F(A\otimes A)}   &&& FA && {F(I)} \\
                 & {F(A)\otimes F(A)} &&&&& I
                    \arrow[""{name=0, anchor=center, inner sep=0}, "{F(\nabla_A)}", from=1-1, to=1-3]
                    \arrow["{\nabla_{FA}}"', from=1-1, to=2-2]
                    \arrow["{\phi_{A,A}}", from=1-3, to=2-2]
                    \arrow[""{name=1, anchor=center, inner sep=0}, "{F(!_A)}", from=1-6, to=1-8]
                    \arrow["{!_{FA}}"', from=1-6, to=2-7]
                    \arrow["{\phi_0}", from=1-8, to=2-7]
                         \arrow["\leq"{marking}, draw=none, from=0, to=2-2]
                    \arrow["\leq"{marking, xshift=-2pt}, draw=none, from=1, to=2-7]
                \end{tikzcd}\]
			\item $F$ is \textbf{colax bicartesian} if it is colax cartesian and colax opcartesian in such a way as to become bilax monoidal (Definition \ref{def:bilax monoidalfunctor}).
		\end{enumerate}
	\end{mydefinition}

\section{Kleisli categories are gs-monoidal}\label{sec:Kleisli and span}

In recent years, strong monads and Kleisli categories have been used to provide categorical models in several branches of computer science.  The leading example is Moggi's work~\cite{Moggi89,Moggi91} on an abstract approach to the notion of computation. We refer to \cite{HCA2,SGL,FOCL} for introductions to the theory of monads and to \cite{IHOCL,PHDTT} for more details. Appendix \ref{sec: strong and commutative monad} offers a short recap of the main definitions. We start by recalling some relevant and known examples of Kleisli categories.

\begin{myexample}\label{ex: powerset monad}
$\Rel$ is the Kleisli category of the powerset monad, where its underlying functor $\freccia{\Set}{P}{\Set}$ sends every set $X$ to its powerset $P(X)$.
\end{myexample}
\begin{myexample}
The category of sets and stochastic maps (with pointwise finite support)
is the Kleisli category of the finite distribution monad $\freccia{\Set}{\mathcal{D}}{\Set}$,
sending a set $X$ to the set $\mathcal{D}(X)$ of finitely supported probability measures on $X$.
\end{myexample}
\begin{myexample}\label{ex: Giry monad}
The category of measurable spaces and Markov kernels is the Kleisli category of the Giry monad $\freccia{\mathbf{Meas}}{\mathcal{G}}{\mathbf{Meas}}$, where $\mathbf{Meas}$ denotes the category of measurable spaces and
measurable functions, see \cite{Giry82,JACOBS2018}.
\end{myexample}

\begin{notation*}
	To avoid confusion between the arrows of $\mA$ and the arrows of a Kleisli category $\mA_T$, we adopt the notation $\freccia{A}{f^\sharp}{T(B)}$ for the representative in $\mA$ of an arrow $\freccia{A}{f}{B}$ in $\mA_T$~\cite{FritzGPR23}.
	We often define a Kleisli arrow $f$ by specifying its representative $f^\sharp$.
	The definition of Kleisli composition then amounts to the equation $(g \circ f)^\sharp = \mu \circ T(g^\sharp) \circ f^\sharp$.
\end{notation*}


Given a monad $(T,\mu,\eta)$ on a symmetric monoidal category $\mA$, it is well known (see e.g.~\cite{Kock70}) that the Kleisli category $\mA_T$ inherits a symmetric monoidal structure precisely when the monad is 
commutative. 
If the base category $\mA$ is cartesian monoidal, then the induced monoidal product in $\mA_T$ may not be cartesian.
A simple example is the powerset monad $\freccia{\Set}{P}{\Set}$ on the category of sets and functions of Example \ref{ex: powerset monad}, whose Kleisli category $\Set_{P}$ is exactly $\Rel$, and the categorical product on $\Set$ induces a monoidal product on $\Rel$ given by the cartesian product. The cartesian product of sets is not the categorical product in $\Rel$, but just a monoidal product. 

Thus a natural question is: what is the algebraic structure that is inherited from the base category 
by the Kleisli category of a commutative monad? Looking at the hierarchy of categories 
sketched in the introduction, spanning from symmetric monoidal to cartesian ones and including gs-monoidal, Markov 
and restriction categories (with restriction products), we answer that gs-monoidality is inherited, while the naturality 
of dischargers or duplicators is not.


%
 
%

\begin{myproposition}\label{prop: Kleisli su cartesiane sono gs}
Let $(T,\mu,\eta)$ be a commutative  monad on a gs-monoidal category $\mA$. Then the Kleisli category $\mA_T$ is a gs-monoidal category with duplicators and dischargers given for every object $A$ by
\[
	\nabla^\sharp_A:=\eta_{A\otimes A}\nabla_A, \qquad !^\sharp_A:=\eta_I!_A.
\]
\end{myproposition}
\begin{proof}

It is known that, under the current assumptions, the Kleisli  category $\mA_T$ is a symmetric monoidal category. Now let us consider an object
 $A$ of the  Kleisli category $\mA_T$.
As mentioned in the statement, we define the arrow $\freccia{A}{\nabla_A}{A\otimes A}$ of $\mA_T$ as represented by the arrow
\[\xymatrix@+1.5pc{
A\ar[r]^{\nabla_A} & A\otimes A \ar[r]^{\eta_{A\otimes A}} & T(A\otimes A)
}\]
of $\mA$. Similarly, we define the arrow $\freccia{A}{!_A}{I}$ of $\mA_T$ as represented by the arrow
\[\xymatrix@+1.5pc{
A\ar[r]^{!_A}& I\ar[r]^{\eta_I} & TI.
}\]
Although it is possible now to verify directly that the axioms of a gs-monoidal category are satisfied, there is a more concise and more insightful argument that works as follows.\footnote{See~\cite[Corollary~3.2]{Fritz_2020}, where this was previously used for Markov categories.}
By definition, the duplicators and dischargers in $\mA_T$ are the images of those in $\mA$ under the inclusion functor $\mA \to \mA_T$, which is strict symmetric monoidal.
It now suffices to note that a strict symmetric monoidal functor maps commutative comonoids to commutative comonoids, and the monoidal multiplicativity conditions~\eqref{monoidal_mult} transfer from $\mA$ to $\mA_T$ for the same reason.
Thus $\mA_T$ is a gs-monoidal category.
\end{proof} 

\begin{myremark}
An equivalent choice for the arrow $\nabla_A^\sharp$ in Proposition~\ref{prop: Kleisli su cartesiane sono gs} is $\nabla_A^\sharp:=c_{A,A}\nabla_{T(A)}\eta_A$, for $\freccia{T(A)\otimes T(A)}{c_{A,A}}{T(A\otimes A)}$ the canonical arrow defined via the commutative structure of the monad $T$, meaning the diagonal of the diagram in Definition \ref{commutative_monad}. It is straightforward to check that $c_{A,A}\nabla_{T(A)}\eta_A=\eta_{A\otimes A}\nabla_A$.
\end{myremark}

\begin{myremark}\label{rem:affine monad}
The above discussed example of the powerset monad on $\mathbf{Set}$ shows that for a commutative monad, the naturality of the
discharger is not preserved in general by the Kleisli category construction (therefore $\mA_T$ is not a restriction category with restriction products in general, even if $\mA$ is), and likewise for the duplicator. So $\mA_T$ is not a Markov category in general either, even if $\mA$ is.

	In fact, if $\mA$ is a Markov category (Example~\ref{ex: Markov categories}) and  $(T,\mu,\eta)$ is a commutative monad on $\mA$, then the monoidal unit $I$ is terminal in $\mA_T$ if and only if $T(I)\cong I$ in $\mA$, a property known as $T$ being an \emph{affine} monad~\cite{Kock71,Jacobs1994}.
	In other words, if the monad $T$ preserves the terminal object, every arrow of $\mA_T$ is $\mA_T$-total and this makes $\mA_T$ into a Markov category~\cite[Corollary~3.2]{Fritz_2020}. As an example, consider the non-empty powerset monad $\freccia{\Set}{P^*}{\Set}$, associating to a set $X$ the family of its non-empty subsets $P(X) \setminus \emptyset$: the arrows of the Kleisli categories are total relations, thus $\total{\Rel}\cong \Set_{P^*}$, and indeed we have that $T(I) \cong I$.
\end{myremark}
\begin{myremark}
For representing functional relations in terms of Kleisli categories, it suffices to consider the lifting monad (also called maybe monad), associating with a set $X$ the pointed set $X_{\bot}:=X+1$.
Its Kleisli category is exactly $\pfn{\Rel}$, the category of sets and partial functions.
\end{myremark}

\begin{myexample}\label{ex:Stoch and QBS_P are gs-monoidal}
The category of measurable spaces and Markov kernels of Example \ref{ex: Giry monad} is gs-monoidal, and actually a Markov category, since the Giry monad $\freccia{\mathbf{Meas}}{\mathcal{G}}{\mathbf{Meas}}$ is an affine commutative monad with respect to the cartesian monoidal structure on $\mathbf{Meas}$.
This Kleisli category is often denoted by $\mathbf{Stoch}$.

Similarly, the category of quasi-Borel spaces $\mathbf{QBS}$ is cartesian and the monad $\freccia{\mathbf{QBS}}{P}{\mathbf{QBS}}$ of probability measures on it is an affine commutative monad, see \cite{Heunen2017} for all the details. Therefore the Kleisli category  $\mathbf{QBS}_P$ is a gs-monoidal category, and in fact a Markov category.
\end{myexample}
Now let us consider a commutative monad $\freccia{\mA}{T}{\mA}$ on a gs-monoidal category $\mA$, and let us denote by $\freccia{T(X)\otimes T(Y)}{c_{X,Y}}{T(X\otimes Y)}$ the canonical lax symmetric monoidal structure defined via the commutative structure  of $T$.

\begin{mydefinition}\label{def:gs-monoidal monad}
	A commutative monad $\freccia{\mA}{T}{\mA}$ is called \textbf{gs-monoidal monad} if $T(\nabla_A)=c_{A,A}\nabla_{T(A)}$ and $T(!_A)=\eta_I!_{T(A)}$ for every object $A$ of $\mA$.
\end{mydefinition}
\begin{myremark}\label{rem: gs-monoidal monad}
	Since $\psi_{X,Y} := c_{X,Y}$ together with $\psi_0 := \eta_I$ makes $T$ into a lax symmetric monoidal functor (Remark~\ref{rem:comm_vs_symmon}), 
	Definition~\ref{def:gs-monoidal monad} can be equivalently introduced as requiring $T$ to be a lax gs-monoidal functor.
\end{myremark}

We have seen in Proposition~\ref{prop: Kleisli su cartesiane sono gs} that Kleisli categories for commutative monads on gs-monoidal categories
inherit the gs-monoidal structure. In the case of gs-monoidal monads on cartesian categories, we obtain the following result.

\begin{mylemma}\label{lem: Kleisli are cartesian}
	Let $\mA$ be a cartesian monoidal category and $\freccia{\mA}{T}{\mA}$ a gs-monoidal monad. Then  $\mA_T$ is cartesian monoidal too.
\end{mylemma}
\begin{proof}
	Let us consider an arrow $\freccia{A}{f}{B}$ of $\mA_T$. We first show that $f$ is functional, meaning $\nabla_B\circ f= (f\otimes f) \circ\nabla_A$.
Reasoning in terms of representing morphisms in $\mA$, observe that 
\[(\nabla_B\circ f)^{\sharp}=\mu_{B\otimes B}T(\eta_{B\otimes B}\nabla_B)f^{\sharp}=T(\nabla_B)f^{\sharp}.\]
By definition of the monoidal structure on $\mA_T$, it is straightforward to check that
\[((f\otimes f) \circ\nabla_A)^{\sharp}=c_{B,B}( f^{\sharp}\otimes f^{\sharp})\nabla_A.\]
Employing first the assumption that $T$ is a gs-monoidal monad and that $\mA$ is cartesian, we have that
\[ T(\nabla_B)f^{\sharp}= c_{B,B}\nabla_{T(B)}f^{\sharp}=  c_{B,B}( f^{\sharp}\otimes f^{\sharp})\nabla_A.\]
Therefore we can conclude that $\nabla_B\circ f=(f\otimes f) \circ\nabla_A$.
Similarly, one can check that $!_Bf={}!_A$ in $\mA_T$.
\end{proof}

\begin{myexample}
	\label{ex:abelian_group}
	Let $G$ be an abelian group.
	Then the functor $\freccia{\Set}{G \times -}{\Set}$ is a commutative monad in a canonical way, and it is easily seen to be a gs-monoidal monad.
	And indeed the resulting monoidal structure on its Kleisli category is cartesian.
\end{myexample}

Note that starting from a Markov category $\mA$ and a gs-monoidal monad $\freccia{\mA}{T}{\mA}$, 
$\mA_T$ is also Markov since the gs-monoidality implies that $!_I$ is inverse to $\eta_I$, making $T$ affine.
Restriction categories with restriction products are similarly preserved.

Now recall that given an arbitrary monad $\freccia{\mA}{T}{\mA}$, it induces a pair of functors $\freccia{\mA}{F_T}{\mA_T}$ and  $\freccia{\mA_T}{G_T}{\mA}$ such that $F_T\dashv G_T$ and $T=G_TF_T$, namely
\begin{itemize}
\item $F_T(X):=X$ and $F_T(f)^\sharp := \eta_Bf$ for $\freccia{A}{f}{B}$ in $\mA$;
\item $G_T(X):=TX$ and $G_T(f):= \mu_BT(f^\sharp) $ for $\freccia{A}{f}{B}$ in $\mA_T$.
\end{itemize}
When $\mA$ is gs-monoidal and $T$ is commutative, the functor $\freccia{\mA}{F_T}{\mA_T}$ is a strict gs-monoidal functor by definition of the gs-monoidal structure on $\mA_T$.
On the other hand, the functor $\freccia{\mA_T}{G_T}{\mA}$ is not strict monoidal in general, but just lax monoidal with structure morphisms $ 	\freccia{T(X)\times T(Y)}{\psi_{X,Y}:=c_{X,Y}}{T(X\times Y)} $ and $\freccia{I}{\psi_0:=\eta_I}{T(I)}$. The functor $G_T$ is lax gs-monoidal precisely when the monad $T$ is so.

\begin{myproposition}
Let $T$ be a commutative monad on a gs-monoidal category $\mA$. Then the Kleisli right adjoint $\freccia{\mA_T}{G_T}{\mA}$ is a lax gs-monoidal functor if and only if $T$ is a gs-monoidal monad.
\end{myproposition}
\begin{proof}
We show first that $\psi$ equips $G_T$ with a lax symmetric monoidal structure, which is true for any commutative monad $T$ on a symmetric monoidal category $\mA$.\footnote{We expect this to be known, but we have not found a precise reference. The closest that we know of is~\cite[Proposition~28]{Guitart80}, which shows the analogous statement for the forgetful functor from the Eilenberg-Moore category $\mA^T$, but under additional assumptions on $T$ needed to make $\mA^T$ monoidal in the first place.}
We start by checking the naturality of $\psi$, i.e.~that for every arrow $\freccia{A_1}{f_1}{B_1}$ and  $\freccia{A_2}{f_2}{B_2}$ of $\mA_T$ the following square commutes
\[\begin{tikzcd}[column sep=7ex]
	{TA_1\otimes  TA_2} &&& {TB_1\otimes TB_2} \\
	\\
	{T(A_1\otimes A_2)} &&& {T(B_1\otimes B_2)}
	\arrow["{\mu_{B_1}T(f_1^{\sharp}) \otimes \,\mu_{B_2}T(f_2^ {\sharp})}", from=1-1, to=1-4]
	\arrow["{c_{A_1,A_2}}"', from=1-1, to=3-1]
	\arrow["{\mu_{B_1\otimes B_2}T(c_{B_1,B_2}(f_1^{\sharp} \otimes f_2^{\sharp}))}"', from=3-1, to=3-4]
	\arrow["{c_{B_1,B_2}}", from=1-4, to=3-4]
\end{tikzcd}\]
First notice that, by naturality of $c$, we have that
\[
	\mu_{B_1\otimes B_2}T(c_{B_1,B_2}(f_1^\sharp \otimes f_2^\sharp))c_{A_1,A_2} = \mu_{B_1\otimes B_2}T(c_{B_1,B_2})c_{TB_1,TB_2}(T(f_1^\sharp)\otimes T(f_2^\sharp)).
\]
The other composite consisting of the upper horizontal arrow followed by the right vertical arrow can also be written as
\[c_{B_1,B_2}(\mu_{B_1}\otimes \mu_{B_2})(T(f_1^\sharp)\otimes T(f_2^\sharp)).\]
Therefore, to conclude the proof of the naturality, it is enough to show that
\[c_{B_1,B_2}(\mu_{B_1}\otimes \mu_{B_2})=\mu_{B_1\otimes B_2}T(c_{B_1,B_2})c_{TB_1,TB_2}.\]
But this equation is exactly the statement that $\mu$ is a monoidal transformation, see e.g.~\cite[Definition~C.2]{FritzPR21}, which proves this statement from the standard equivalence between commutative monads and symmetric monoidal monads.
Hence $\psi$ provides a natural transformation.
Similarly, it is straightforward but tedious to check that the associativity and unitality axioms of a lax symmetric monoidal functor (see Definition~\ref{def:lax monoidal functor}) are satisfied.

Now, notice that, by definition,  $\freccia{\mA_T}{G_T}{\mA}$ is a lax gs-monoidal functor if and only if the following diagrams commute for all $A$ in $\mA_T$
\[\begin{tikzcd}[column sep=tiny]
            {T(A)} && {T(A\otimes A)}   &&& TA &&&& {T(I)} \\
            & {T(A)\otimes T(A)} &&&&&& I
            \arrow[ "{\mu_{A\otimes A}T(\eta_{A\otimes A}\nabla_A)}", from=1-1, to=1-3]
            \arrow["{\nabla_{TA}}"', from=1-1, to=2-2]
            \arrow["{c_{A,A}}"', from=2-2, to=1-3]
            \arrow[ "{\mu_IT(\eta_I!_A)}", from=1-6, to=1-10]
            \arrow["{!_{TA}}"', from=1-6, to=2-8]
            \arrow["{\eta_I}"', from=2-8, to=1-10]
\end{tikzcd}\]
 Since $\mu_{A\otimes A}T(\eta_{A\otimes A})=\id_{T(A \otimes A)}$ and $\mu_IT(\eta_I) = \id_{T(I)}$, these two diagrams commute exactly when the monad $T$ is a gs-monoidal monad.
\end{proof}

\begin{myremark}\label{rem: gs monads on cartesian categories}
	If $T$ is a gs-monoidal monad on a cartesian category $\mA$, then $\mA_T$ is cartesian monoidal by Lemma~\ref{lem: Kleisli are cartesian}, but $G_T:\mA_T\to\mA$ is generally just a lax gs-monoidal functor that is not strong gs-monoidal (see Example~\ref{ex:abelian_group}).
\end{myremark}
\subsection{The oplax cartesian structure of Kleisli categories}
Now we present an ``oplax cartesian version''  of Proposition~\ref{prop: Kleisli su cartesiane sono gs}.
To achieve this goal, we introduce the following generalization of Definition~\ref{def:gs-monoidal monad}.
\begin{mydefinition}\label{def:olax cartesian monad}
Let $\mA$ be an oplax cartesian category. A preorder-enriched commutative monad $\freccia{\mA}{T}{\mA}$ is called \textbf{colax cartesian monad} if $T(\nabla_A)\leq c_{A,A}\nabla_{TA}$ and $T(!_A)\leq \eta_I!_{TA}$ for every object $A$ of $\mA$.
\end{mydefinition}
So in terms of the lax symmetric monoidal structure of $T$ of Remark~\ref{rem: gs-monoidal monad}, a preorder-enriched commutative monad is a colax cartesian monad if and only if the underlying functor is colax cartesian (see Definition~\ref{def:(op) oplax cartesian functor}).

\begin{myproposition}\label{prop: Kleisli is oplax cartesian}
Let $\mA$ be an oplax cartesian category and $\freccia{\mA}{T}{\mA}$ a colax cartesian monad. Then the Kleisli category $\mA_T$ equipped with the preorder given by
\begin{equation}
	\label{eq:preorder on Kleisli}
	f\leq_{\mA_T} g \iff f^{\sharp}\leq_{\mA} g^{\sharp}
\end{equation}
is oplax cartesian as well, and the functors
\[
	\freccia{\mA}{F_T}{\mA_T}, \qquad
	\freccia{\mA_T}{G_T}{\mA}
\]
are colax cartesian (with $F_T$ even in the strict sense).
\end{myproposition}


\begin{proof}
\begin{enumerate}
	\item It is direct to check that $\mA_T$ is preorder-enriched with the preorder given by~\eqref{eq:preorder on Kleisli}, based on the fact that $\mA$ is preorder-enriched and that $T$ is a preorder-enriched functor.
		Similarly for the monotonicity of $\otimes$.
	
	Checking the oplax cartesianity of $\mA_T$ is analogous to the proof of Lemma~\ref{lem: Kleisli are cartesian}.
	So consider an arrow $\freccia{A}{f}{B}$ of $\mA_T$. We first show $\nabla_B\circ f\leq_{\mA_T} (f\otimes f) \circ\nabla_A$. By definition of $\leq_{\mA_T}$, this holds if and only if $(\nabla_B\circ f)^{\sharp}\leq_{\mA} ((f\otimes f) \circ\nabla_A)^{\sharp}$. Note that 
\[(\nabla_B\circ f)^{\sharp}=\mu_{B\otimes B}T(\eta_{B\otimes B}\nabla_B)f^{\sharp}=T(\nabla_B)f^{\sharp}.\]
By definition of the monoidal structure on $\mA_T$, it is direct to check that
\[((f\otimes f) \circ\nabla_A)^{\sharp}=c_{B,B}( f^{\sharp}\otimes f^{\sharp})\nabla_A.\]
Employing first the assumption that $T$ is a colax cartesian functor and then that $\mA$ is oplax cartesian, we have that
\[ T(\nabla_B)f^{\sharp}\leq c_{B,B}\nabla_{T(B)}f^{\sharp}\leq c_{B,B}( f^{\sharp}\otimes f^{\sharp})\nabla_A.\]
Therefore we can conclude that $\nabla_B\circ f\leq_{\mA_T} (f\otimes f) \circ\nabla_A$.
Similarly, one can check that $!_Bf\leq {}!_A$ in $\mA_T$.
	
	\item $F_T$ is trivially colax cartesian since $F_T(\nabla_A) = \nabla_{F_T(A)}$ and $F_T(!_A) = {}!_{F(A)}$ by definition of the gs-monoidal structure on $\mA_T$.
	\item $G_T$ is colax cartesian by $G_T(\nabla_A) = \mu_{A \otimes A} \circ T(\nabla_A^\sharp) = T(\nabla_A)$ and the assumption that $T$ is a colax cartesian functor.
		\qedhere
\end{enumerate}
\end{proof}

\section{Span categories are oplax cartesian}
\label{sec:span}

Proposition~\ref{prop: Kleisli su cartesiane sono gs} shows that the Kleisli category $\mA_T$  of a commutative monad  $T$ on a 
gs-monoidal category $\mA$ is gs-monoidal.
Generalising previous results developed for specific categories (see e.g.~\cite{Bruni2003}),
we show now that for any category $\mA$ with finite limits, the category $\mathbf{PSpan}(\mA)$, 
obtained by taking the preorder reflection of the 2-cells of the bicategory of spans
and identifying two arrows when they are isomorphic spans, is oplax cartesian.

Furthermore, we have the following comparison with Kleisli categories.
By composing with the functor $\freccia{\mA_T}{G_T}{\mA}$, we obtain 
a lax gs-monoidal functor $\mA_T \to \mathbf{PSpan}(\mA)$,
which moreover in many examples is actually faithful.

Let us start with the following standard definition.

\begin{mydefinition}
Let $\mA$ be a category with pullbacks. Then the \emph{bicategory of spans} $\mathbf{Span}(\mA)$ has the same objects 
as $\mA$, and its arrows are defined as
\begin{itemize}
	\item an arrow from $X$ to $Y$ is a \emph{span} $(X  \leftarrow A \to Y)$ of $\mA$;
	\item the identity of $X$ is the span $X \xleftarrow{\id_X}X  \xrightarrow{\id_X} X$;
	\item the composition of spans $X\leftarrow  A \xrightarrow{f} Y$ and $Y\xleftarrow{g}  B \rightarrow Z$ is given by the span $X\leftarrow A\times_{f,g}B \to Z$ obtained by taking the pullback of $f$ and $g$;
	\item a 2-cell $\alpha: (X \leftarrow A \to Y) \Rightarrow (X \leftarrow  B \to Y)$ is an arrow $\alpha: A \to B$ in ${\mA}$  such that the following diagram commutes
		\[
			\begin{tikzcd}[row sep=2pt, column sep=8pt]
				& A \ar[dl, bend right] \ar[rd, bend left]\ar[dd, "\alpha"] & \\
				X  & & Y  \\
				& B \ar[ul, bend left] \ar[ur,  bend right]
			\end{tikzcd}
		\]
	\item vertical composition of 2-cells is given by composition in $\mA$;
	\item horizontal composition of 2-cells as well as associators and unitors are induced by the universal property of pullbacks.
\end{itemize}
\end{mydefinition}

\begin{mydefinition}
Let $\mA$ be a category with pullbacks and $\mathbf{Span}(\mA)$ its bicategory of spans. 
Then the preorder-enriched category $\mathbf{PSpan}(\mA)$ has
\begin{itemize}
\item the same objects as $\mA$;
\item isomorphism classes of arrows of $\mathbf{Span}(\mA)$ as arrows: spans $(X \xleftarrow{f} A  \xrightarrow{g} Y)$ and  $(X \xleftarrow{f'} A'  \xrightarrow{g'} Y)$ are isomorphic if there is an iso $\freccia A i {A'}$ such that $f'\circ i = f$ and $g'\circ i = g$; 
\item a preorder enrichment defined as  $[(X \xleftarrow{f} A  \xrightarrow{g} Y)] \leq [(X \xleftarrow{f'} A'  \xrightarrow{g'} Y)]$ if there is a 2-cell
	\[
		\alpha:  (X \xleftarrow{f} A  \xrightarrow{g} Y) \Rightarrow (X \xleftarrow{f'} A'  \xrightarrow{g'} Y)
	\]
	in $\mathbf{Span}(\mA)$.
\end{itemize}
\end{mydefinition}


It is straightforward to see that $\mathbf{PSpan}(\mA)$ is indeed a preorder-enriched category.
Note that it is locally small as soon as $\mA$ is small.

Now, when $\mA$ is also a cartesian category, it is direct to check that the categorical product $\times$ of $\mA$ induces a monoidal product $\otimes$ on $\mathbf{PSpan}(\mA)$. However, we have actually more structure, as witnessed by the following result.

\begin{myproposition}\label{prop: Pspan(A) is oplax cartesian}
Let $\mA$ be a category with finite limits. Then $\mathbf{PSpan}(\mA)$ is an oplax cartesian category
with
\[
	\nabla^s_X = (X \xleftarrow{\id} X \xrightarrow{\nabla_X} X \times X), \qquad !^s_X = (X \xleftarrow{\id} X \xrightarrow{!} 1).
\]
\end{myproposition}
\begin{proof}
	It is well-known~\cite[Section~3.1]{Bruni2003}\footnote{Even if that reference just considers spans in $\Set$, the proofs work for any category with finite limits.} and also easy to check in the same way as in Proposition~\ref{prop: Kleisli su cartesiane sono gs} that $\mathbf{PSpan}(\mA)$ is gs-monoidal with respect to the given duplicators and dischargers.
	So we just verify the axioms for oplax cartesianity in addition.
	Let us consider an arrow from $X$ to $Y$ in $\mathbf{PSpan}(\mA)$ represented by a span $(X \xleftarrow{f} A \xrightarrow{g} Y)$. We have to show the inequality
	\begin{align*}
		\nabla_Y^s & \circ(X \xleftarrow{f} A \xrightarrow{g} Y)  \le ((X \xleftarrow{f} A \xrightarrow{g} Y)\otimes (X \xleftarrow{f} A \xrightarrow{g} Y) )\circ\nabla_X^s.
	\end{align*}
First, note that by definition of composition in $\mathbf{PSpan}(\mA)$, we have that
$$\nabla_Y^s\circ(X \xleftarrow{f} A \xrightarrow{g} Y)=(X \xleftarrow{f} A \xrightarrow{\nabla_Y g} Y\times Y),
$$
and since $\mA$ is cartesian, and hence $\nabla_Yg=(g\times g)\nabla_A$, this evaluates further to
$$\nabla_Y^s\circ(X \xleftarrow{f} A \xrightarrow{g} Y)=(X \xleftarrow{f} A \xrightarrow{(g\times g)\nabla_A} Y\times Y).
$$
Now, employing the universal property of pullbacks and the naturality of $\nabla$ in $\mA$, it is direct to check that 
\begin{align*}
	(X \xleftarrow{f} & A \xrightarrow{(g\times g)\nabla_A} Y\times Y) \leq ((X \xleftarrow{f} A \xrightarrow{g} Y)\otimes(X \xleftarrow{f} A \xrightarrow{g} Y))\circ\nabla_X^s,
\end{align*}
as was to be shown.
Similarly we have the inequality
$$!_Y^s\circ(X \xleftarrow{f} A \xrightarrow{g} Y)\leq {}!_X^s$$
via the 2-cell obtained via $f$, since the left-hand side is equal to $(X \xleftarrow{f} A \xrightarrow{!_A} 1)$.
\end{proof}

We may characterise weak functionality and weak totality in terms of
properties of the components of a span.

\begin{myproposition}\label{prop_functional_and_total_spans}
	Let $ (X \xleftarrow{f} A \xrightarrow{g}Y )$ be an arrow in $\mathbf{PSpan}(\mA)$. Then
	\begin{enumerate}
	\item it is weakly $\mathbf{PSpan}(\mA)$-functional if and only if for every $\freccia{Z}{h_1,h_2}{A}$ we have that
	$fh_1=fh_2$ implies $gh_1=gh_2$;
	\item it is weakly $\mathbf{PSpan}(\mA)$-total if and only if $f$ is a split epimorphism.
	\end{enumerate}
\end{myproposition}
The first item generalises the intuition behind the use of spans with mono left leg
for modelling partial functions.
\begin{proof}
	\begin{enumerate}
	\item By the first axiom of oplax cartesian categories, it is enough to show that the inequality
		\begin{align*}
			 (X\times X \xleftarrow{f\times f} A \times A\xrightarrow{g\times g}Y\times Y )\circ \nabla_X^s \leq \nabla_Y^s & \circ(X \xleftarrow{f} A \xrightarrow{g}Y )	
		\end{align*}
		holds if and only if $fh_1=fh_2$ implies $gh_1=gh_2$. As in the previous proof, we have
	$$\nabla_Y^s\circ(X \xleftarrow{f} A \xrightarrow{g}Y )=(X \xleftarrow{f} A \xrightarrow{(g \times g)\nabla_A}Y\times Y ), $$
	while $(X\times X \xleftarrow{f\times f} A \times A\xrightarrow{g\times g}Y\times Y )\circ \nabla_X^s$ is given by the composite span in

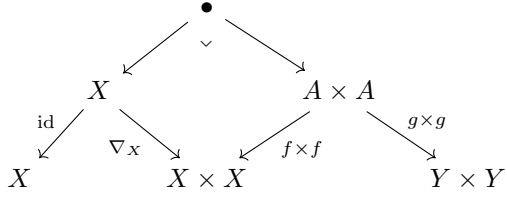
\[\begin{tikzcd}[column sep=small]
	&& \bullet \\
	& X && {A\times A} \\
	X && {X\times X} && {Y\times Y}
	\arrow["\id"', from=2-2, to=3-1]
	\arrow["{\nabla_X}"', from=2-2, to=3-3]
	\arrow[from=1-3, to=2-2]
	\arrow[from=1-3, to=2-4]
	\arrow["{f\times f}", from=2-4, to=3-3]
	\arrow["{g\times g}", from=2-4, to=3-5]
	\arrow["\lrcorner"{anchor=center, pos=0.125, rotate=-45}, draw=none, from=1-3, to=3-3]
\end{tikzcd}\]

Thus, by the universal property of pullbacks and the definition of the preorder $\leq$ in $\mathbf{PSpan}(\mA)$, it is direct to check that $fh_1=fh_2$ implies $gh_1=gh_2$ if and only
\begin{align*}
			 (X\times X \xleftarrow{f\times f} A \times A\xrightarrow{g\times g}Y\times Y )\circ \nabla_X^s\leq \nabla_Y^s & \circ(X \xleftarrow{f} A \xrightarrow{g}Y )
		\end{align*}
i.e. if and only if $(X \xleftarrow{f} Z \xrightarrow{g}Y )$ is weakly $\mathbf{PSpan}(\mA)$-functional. 

\item Notice that $!_Y^s\circ(X \xleftarrow{f} A \xrightarrow{g}Y ) = (X \xleftarrow{f} A \xrightarrow{!_A}1 )$. Hence the relevant inequality
	\[
		!_X^s \leq !_Y^s\circ(X \xleftarrow{f} A \xrightarrow{g}Y ) 
	\]
	holds if and only if there exists an arrow $\freccia{X}{h}{A}$ such that $fh=\id_X$, i.e.~if and only if $f$ is a split epimorphism. 
	\qedhere
\end{enumerate}
\end{proof}

An immediate corollary of Proposition~\ref{prop_functional_and_total_spans} follows.
\begin{mycorollary}
Let $f$ be an isomorphism in $\mA$.
Then the span $ (X \xleftarrow{f} Z \xrightarrow{g}Y )$ of $\mathbf{PSpan}(\mA)$ is weakly $\mathbf{PSpan}(\mA)$-functional and  
weakly $\mathbf{PSpan}(\mA)$-total.
\end{mycorollary}

These morphisms are those that are in the image of the canonical inclusion functor $\mA \to \mathbf{PSpan}(\mA)$.	

Finally, if we consider a gs-monoidal monad $T$ on a category with finite limits $\mA$, we can employ Remark~\ref{rem: gs monads on cartesian categories} to show the following result.
\begin{myproposition}
Let $\freccia{\mA}{T}{\mA}$ be a gs-monoidal monad on a category with finite limits $\mA$. Then the canonical functor $\mA_T\to \mathbf{PSpan}(\mA)$, obtained by composing the right adjoint $G_T:\mA_T\to \mA$ and the canonical inclusion $\mA\to \mathbf{PSpan}(\mA)$, is a lax gs-monoidal functor.
\end{myproposition}

\section{On functorial completeness}
\label{sec:completeness}
In this section, we first present a gs-monoidal Yoneda embedding, and then 
a functorial completeness result for oplax cartesian categories.
%

\subsection{A gs-monoidal Yoneda embedding}

For $\mC$ a symmetric monoidal category, let us consider functors $F : \mC \to \Set$.
Assuming that $\mC$ is small, these functors form a symmetric monoidal category with respect to \textbf{Day convolution} $\boxtimes$~\cite{Day70b,Day70}, where for $X \in \mC$
\[
	(F \boxtimes G)(X) \coloneqq \int^{A,B \in \mC} \mC(A \otimes B, X) \times F(A) \times G(B),
\]
and $F \boxtimes G$ is defined on arrows in terms of the universal property of the coend.
$F \boxtimes G$ enjoys a universal property, which states that the natural transformations $F \boxtimes G \to H$ for any functor $H : \mC \to \Set$ are in natural bijection with the transformations
\[
	\begin{tikzcd}[column sep=2pc]
		F(X) \times G(Y) \ar{r}		& H(X \otimes Y)
	\end{tikzcd}
\]
natural in $X, Y \in \mC$.
Defined like this, Day convolution turns the category of functors $\mC \to \Set$ into a symmetric monoidal category.
The associator is obvious, and with the monoidal unit given by the hom-functor $\mC(I,-)$, the left unitor component $\mC(I,-) \boxtimes F \to F$ at any $F$ corresponds to the transformation with components
\[
	\begin{tikzcd}[row sep=1pt,column sep=2pc]
		\mC(I,X) \times F(Y) \ar{r}	& F(X \otimes Y)	\\
		(f, \alpha) \ar[mapsto]{r}	& F(f \otimes \id_Y)(\alpha)
	\end{tikzcd}
\]
and similarly for the right unitor.
The braidings are inherited from $\mC$.
We denote by $\mathbf{LaxSymMon}(\mC,\Set)$ the category of lax symmetric monoidal functors $\mC \to \Set$ and \emph{all} natural transformations.

\begin{mylemma}
	\label{lax_is_cogs}
	Let $\mC$ be a small monoidal category. Then the category of lax symmetric monoidal functors $\mathbf{LaxSymMon}(\mC,\Set)^{\op}$ is a gs-monoidal category in a canonical way.
	Its total and functional arrows are exactly the formal opposites of the monoidal natural transformations.
\end{mylemma}

\begin{proof}
	For any $F : \mC \to \Set$, there is an equivalence between lax symmetric monoidal structures on $F$ and commutative monoid structures with respect to Day convolution, in such a way that the monoidal natural transformations are in natural bijection with the monoid homomorphisms\footnote{See e.g.~\cite[Example~3.2.2]{Day70} or~\cite[Proposition~22.1]{Mandell01} for a version of the statement for presheaves with values in topological spaces.}.
	Thus it suffices to show that the category of commutative monoids in the symmetric monoidal category of functors under Day convolution is 
	co-gs-monoidal in a canonical way.
	But this latter statement is an instance of the fact that the category of commutative monoids in \emph{any} symmetric monoidal category is a co-gs-monoidal category in a canonical way (Example~\ref{comon_is_gs}).

	In any such category, the co-total and co-functional arrows are exactly the monoid homomorphisms.
	This implies the claim that the total and functional arrows in $\mathbf{LaxSymMon}(\mC,\Set)^{\op}$ are exactly the formal opposites of monoidal natural transformations.
\end{proof}

We think of the categories $\mathbf{LaxSymMon}(\mC,\Set)^{\op}$ as gs-monoidal analogues of the functor categories in the usual Yoneda lemma.
The gs-monoidal Yoneda embedding then reads as follows.

\newcommand{\yo}{\mathcal{Y}}
\begin{myproposition}\label{prop:Yoneda}
	Let $\mC$ be a small gs-monoidal category. Then there is a fully faithful oplax gs-monoidal functor
	\[
		\begin{tikzcd}[column sep=1.4pc]
			\yo : \mC \ar{r}	& \mathbf{LaxSymMon}(\mC,\Set)^{\op}
		\end{tikzcd}
	\]
\end{myproposition}

\begin{proof}

	On objects, we define $\yo(A) \coloneqq \mC(A,-)$, which is a lax monoidal functor $\mC \to \Set$; for the lax monoidal structure,
	  we refer forward to the proof of Theorem~\ref{thm:simple_complete_bi_lax}.
	The action of $\yo$ on arrows $f : B \to A$ is by precomposition, and it defines a natural transformation 
	$\mC(B,-) \to \mC(A,-)$.
	Full faithfulness of $\yo$ holds by the standard Yoneda embedding.

	It remains to equip $\yo$ with a oplax gs-monoidal structure, recalling that $\mathbf{LaxSymMon}(\mC,\Set)^{\op}$ carries the gs-monoidal structure introduced in Lemma~\ref{lax_is_cogs}.
	For the oplaxator, note that we have a transformation
	\[
		\begin{tikzcd}[row sep=1pt,column sep=1.4pc]
			\mC(A,X) \otimes \mC(B,Y) \ar{r}	& \mC(A \otimes B, X \otimes Y)	\\
			(f,g)			\ar{r}		& f \otimes g
		\end{tikzcd}
	\]
	that is natural in all four arguments $A,B,X,Y \in \mC$. By the universal property of Day convolution, this can be regarded as an arrow
	\[
		\begin{tikzcd}[column sep=1.4pc]
			\mC(A,-) \boxtimes \mC(B,-) \ar{r}	& \mC(A \otimes B,-)
		\end{tikzcd}
	\]
	in $\mathbf{LaxSymMon}(\mC,\Set)$.
	Its naturality in $A$ and $B$ also follows by the universal property.
	Let us show that considering these transformations in the opposite category defines the comultiplication of the claimed oplax monoidal structure on $\yo$, while by $\yo(I) = \mC(I,-)$ we have strict counitality.
	The coassociativity equation for the comultiplication amounts exactly to the associativity of the monoidal structure of $\mC$.
	The left unitality equation holds by the commutativity of the diagram
	\[
		\begin{tikzcd}[column sep=small]
			&	\mC(I,-) \boxtimes \mC(A,-) \ar{dr}{\text{\footnotesize{laxator}}}	\\
			\mC(A,-) \ar{ur}{\parbox{2.2cm}{\centering \footnotesize{unitor of\\ Day convolution}}} \ar[swap]{rr}{\text{\footnotesize{induced by unitor in }} \mC}	&& \mC(I \otimes A,-)
		\end{tikzcd}
	\]
	and similarly for right unitality.

	The preservation of duplicators amounts to the diagram
	\[
		\begin{tikzcd}[column sep=small]
			& \mC(A,-) \boxtimes \mC(A,-) \ar{dr}{\text{\footnotesize{laxator}}} \ar[swap]{dl}{\text{\footnotesize{lax structure on }} \mC(A,-)}	\\
			\mC(A,-) && \mC(A \otimes A, -) \ar[swap]{ll}{\mC(\nabla_A,-)}	
		\end{tikzcd}
	\]
	which holds by definition of the lax structure on $\mC(A,-)$, and similarly for the dischargers.
\end{proof}

\subsection{Functorial completeness for oplax cartesian categories}\label{sec:Bilax-completeness}

In this section, we consider the category $\Preord$ of preordered sets and monotone maps as a preorder-enriched cartesian monoidal category, and therefore in particular an oplax cartesian category.

Here we start considering the problem of extablishing a functorial completeness result for oplax cartesian categories with respect to the oplax cartesian category $\Preord$. The following theorem, in addition to demonstrating completeness for colax bicartesian functors, will be crucial for our subsequent results: in particular, we will use this theorem to prove Theorem~\ref{thm:rel_complete_bi_lax} where we show a completeness result (for suitable functors) in $\Rel$.
\begin{mytheorem}[bilax completeness to $\Preord$]
	\label{thm:simple_complete_bi_lax}
	Let $\mC$ be a locally small oplax cartesian category $\mC$ and $f, g : X \to Y$ arrows in $\mC$. Then we have
	\begin{enumerate}
		\item $f\leq g$ if and only if $F(f)\leq F(g)$ for every colax bicartesian functor $F : \mC \to \Preord$;
		\item $f \approx g$ if and only if $F(f) \approx F(g)$ for every colax bicartesian functor $F : \mC \to \Preord$.
	\end{enumerate}
\end{mytheorem}

Since every colax bicartesian functor is in particular colax cartesian and colax opcartesian, the same statements hold for these classes of functors as well.

\begin{proof}
	By the usual Yoneda lemma, it is enough to show that every hom-functor $\mC(A,-) : \mC \to  \Preord$ has a canonical colax bicartesian structure.

	First we show that every hom-functor $\mC(A,-) : \mC \to  \Preord$ is colax cartesian.
	The lax symmetric monoidal structure, in components,
	is given by 
	\begin{equation}
		\label{eq:laxator_homfunctor}
		\psi_{X,Y} \: :
		\begin{tikzcd}[row sep=1pt,column sep=1.4pc]
			\mC(A,X) \times \mC(A,Y) \ar{r}		& \mC(A,X \otimes Y)	\\
			(f, g)\ar[mapsto]{r}			& (f \otimes g) \circ \nabla_A,
		\end{tikzcd}
	\end{equation}
	and
	\begin{equation}
		\label{eq:laxator_homfunctor2}
		\psi_0 \: :
		\begin{tikzcd}[row sep=1pt,column sep=1.4pc]
			1  \ar{r}		& \mC(A,I )	\\
			\bullet \ar[mapsto]{r}			& !_A.
		\end{tikzcd}
	\end{equation}
	It is straightforward to verify that $\psi_{X,Y}$ is natural in $X$ and $Y$ and satisfies the relevant coherences.
	The preservation of the braiding holds by the commutativity assumption on $\nabla_A$.
	This makes $\mC(A,-)$ lax symmetric monoidal.
	The claim that it is colax cartesian then amounts to the inequalities
	\begin{equation}
		\label{eq:gs_homfunctor}
                \begin{tikzcd}[column sep=tiny]
                 {\mC(A,X)} && {\mC(A,X \otimes X)} \\
                 & { \mC(A,X) \times \mC(A,X)}
                    \arrow[""{name=0, anchor=center, inner sep=0}, "{\mC(A,\nabla_X)}", from=1-1, to=1-3]
                    \arrow["{\nabla_{\mC(A,X)}}"', from=1-1, to=2-2]
                 \arrow["{\psi_{X,Y}}"', from=2-2, to=1-3]
                           \arrow["\leq"{marking}, draw=none, from=0, to=2-2]
                \end{tikzcd}
	\end{equation}
	and
	\begin{equation}
		\label{eq:gs_homfunctor2}
               \begin{tikzcd}[column sep=tiny]
                 {\mC(A,X)} &&&& {\mC(A,I)} \\
                 && {1}
                    \arrow[""{name=0, anchor=center, inner sep=0}, "{\mC(A,!_X)}", from=1-1, to=1-5]
                    \arrow["{!_{\mC (A,X)}}"', from=1-1, to=2-3]
                 \arrow["{\psi_0}"', from=2-3, to=1-5]
                           \arrow["\leq"{marking}, draw=none, from=0, to=2-3]
                \end{tikzcd}
	\end{equation}
	which hold since they are the defining inequalities of the colax cartesianity of $\mC$.

	Now we show that every hom-functor $\mC(A,-) : \mC \to  \Preord$ is colax opcartesian.
	The oplax symmetric monoidal structure, in components,
	is given by 
	\[
		\phi_{X,Y} \: :
		\begin{tikzcd}[row sep=1pt,column sep=1.4pc]
			\mC(A,X\otimes Y) \ar{r}		& \mC(A,X)\times \mC(A,Y)	\\
			f \ar[mapsto]{r}			& ((\id_X\otimes \,!_Y)\circ f, (!_X\otimes \id_Y) \circ f),
		\end{tikzcd}
	\]
	and
	\[
		\phi_0 \: :
		\begin{tikzcd}[row sep=1pt,column sep=1.4pc]
			\mC(A,I )  \ar{r}		& 1	\\
			f \ar[mapsto]{r}			& \bullet.
		\end{tikzcd}
	\]
	It is straightforward to check that $\phi_{X,Y}$ is natural in $X$ and $Y$ and that the coherence axioms for colax monoidal functors are satisfied.
	This makes $\mC(A,-)$ colax opcartesian since the inequalities
	\[
               \begin{tikzcd}[column sep=tiny]
                 {\mC(A,X)} && {\mC(A,X \otimes X)} \\
                 & { \mC(A,X) \times \mC(A,X)}
                    \arrow[""{name=0, anchor=center, inner sep=0}, "{\mC(A,\nabla_X)}", from=1-1, to=1-3]
                    \arrow["{\nabla_{\mC(A,X)}}"', from=1-1, to=2-2]
                    \arrow["{\phi_{X,Y}}", from=1-3, to=2-2]
                           \arrow["\leq"{marking}, draw=none, from=0, to=2-2]
                \end{tikzcd}
	\]
and
	\[
               \begin{tikzcd}[column sep=tiny]
                 {\mC(A,X)} &&&& {\mC(A,I)} \\
                 && {1}
                    \arrow[""{name=0, anchor=center, inner sep=0}, "{\mC(A,!_X)}", from=1-1, to=1-5]
                    \arrow["{!_{\mC(A,X)}}"', from=1-1, to=2-3]
                    \arrow["{\phi_0}", from=1-5, to=2-3]
                           \arrow["\leq"{marking}, draw=none, from=0, to=2-3]
                \end{tikzcd}
	\]
	hold even with equality.

	Finally, we show that $\mC(A,-)$ is a bilax monoidal functor in the sense of Definition \ref{def:bilax monoidalfunctor}.
	We first check the braiding axiom. Let us consider two arrows $\freccia{A}{f}{W\otimes X}$ and $\freccia{A}{g}{Y\otimes Z}$ of $\mC$. We have to show that 
	\begin{equation}
		\label{eq_left_leg_hex}
		\phi_{W\otimes Y,X\otimes Z}\circ \mC(A,\id_W\otimes \gamma \otimes \id_Z)\circ \psi_{W\otimes X,Y\otimes Z}(f,g)
	\end{equation}
	is equal to 
	\begin{equation}
		\label{eq_right_leg_hex}
		(\psi_{W,Y}\times \psi_{X,Z})\circ (\id_W\times \gamma \times \id_Z)\circ (\phi_{W,X}\times \phi_{Y,Z})(f,g).
	\end{equation}
	Straightforward evaluation of both sides shows that indeed both \eqref{eq_left_leg_hex} and \eqref{eq_right_leg_hex} are equal to the element 
	\begin{align*}
		( (\id_{W}\otimes \,!_X\otimes & {}\id_{Y}\otimes \,!_{Z})(f\otimes g)\nabla_A, \\
		& ({}!_{W} \otimes {}\id_{X}\otimes \,!_{Y} \otimes {}\id_Z)(f\otimes g)\nabla_A)
	\end{align*}
	of $\mC(A,W\otimes Y)\times \mC(A,X\otimes Z)$. Finally, we check the unitality axioms. The first axiom follows from the fact that both the arrows
	\[
		\freccia{1}{\phi_{I,I}\circ \mC(A,\lambda_I)\circ \psi_0}{\mC(A,I)\times \mC(A,I)}
	\]
	and 
	\[
		\freccia{1}{(\psi_0\otimes \psi_0)\circ \lambda_I}{\mC(A,I)\times \mC(A,I)}
	\]
	act as $\bullet\mapsto (!_A,!_A)$. Moreover, the second and third axiom hold because $1$ is terminal in $\Preord$.
\end{proof}

\begin{myremark}
	As we wrote in the proof above, by the usual Yoneda lemma, Theorem~\ref{thm:simple_complete_bi_lax} boils down to showing that every hom-functor $\mC(A,-) : \mC \to  \Preord$ has a canonical colax bicartesian structure. However, the same proof \emph{cannot} be adapted to a functorial completeness result for gs-monoidal categories and bilax gs-monoidal functors $\mC \to \Set$. Indeed, for a gs-monoidal category $\mC$, the hom-functor $\mC(A,-) : \mC \to  \mathbf{Set}$ is a lax gs-monoidal functor with respect to the laxator of~\eqref{eq:laxator_homfunctor} and~\eqref{eq:laxator_homfunctor2} if and only if $\mC$ is cartesian, since for gs-monoidality the diagrams~\eqref{eq:gs_homfunctor} and~\eqref{eq:gs_homfunctor2} would have to commute on the nose, implying that every arrow is total and functional.
\end{myremark}

We conclude by showing how the completeness result of Theorem~\ref{thm:simple_complete_bi_lax} for colax bicartesian functors into $\Preord$ can be transferred to a completeness result with respect to a suitable class of mappings into $\Rel$. 

To achieve this goal, we first define the new class of functors we want to consider: let us say that a mapping $\freccia{\mC}{F}{\mD}$ between preorder-enriched categories is a \textbf{lax-on-identities functor} if it is like a preorder-enriched functor with strict preservation of binary composition, but where identities are preserved \emph{only laxly} in the sense that
\[
	\id_{F(A)} \le F(\id_A)
\]
for all objects $A$ in $\mC$.
If $\freccia{\mC}{F}{\mD}$ is a mapping between preorder-enriched monoidal categories, then we say that it is a \textbf{lax monoidal lax-on-identities functor} if it is a lax-on-identities functor together with transformations $\psi_{X,Y}$ and $\psi_0$, as for a lax monoidal functor, such that the associativity and unitality diagrams in Definition~\ref{def:lax monoidal functor} commute upon replacing every occurrence of an identity $\id_{F(A)}$ by $F(\id_A)$. Note that this implies, for example, that the original unitality diagrams~\eqref{eq:lax_monoidal_unitality} commute only laxly, meaning with $\le$ from top to bottom.
\begin{mylemma}\label{lem:compositin_lax_mon_and_lax_on_ident_mon}
	For preorder-enriched monoidal categories $\mA$, $\mC$ and $\mD$, let $\freccia{\mA}{G}{\mC}$ be a lax monoidal functor and $\freccia{\mC}{F}{\mD}$ a lax monoidal lax-on-identities functor.
	Then the composition $\freccia{\mA}{F\circ G}{\mD}$ is a lax monoidal lax-on-identities functor.
\end{mylemma}
\begin{proof}
Let us denote by $\psi^F$ and $\psi_0^F$ the lax monoidal structure of $F$ and by $\psi^G$ and $\psi_0^G$ the one of $G$. It is straightforward to check that $F\circ G$ is a lax monoidal lax-on identities functor with $\psi_{X,Y}^{FG}$ given by the composite
\[\begin{tikzcd}[column sep=small]
	{FG(X)\otimes FG(Y)} && {F(G(X)\otimes G(Y))} && {FG(X\otimes Y)}
	\arrow["{\psi_{GX,GY}^F}", from=1-1, to=1-3]
	\arrow["{F(\psi^G_{X,Y})}", from=1-3, to=1-5]
\end{tikzcd}\]
and $\psi_0^{FG}:=F(\psi_0^G) \psi_0^F$.
In particular, it is easy to check that $F\circ G$ is a lax-on-identities functor since
\[
	\id_{FG(A)}\leq F(\id_{G(A)})=FG(\id_A),
\]
and $\psi^{FG}$ a natural transformation. Finally, it is straightforward to show that $\psi^{FG}$ satisfies the associativity axiom of lax monoidal lax-on-identity functors, i.e. for all objects $A,B,C$ in $\mC$
\begin{align*}
	\label{eq:FG_assoc}
	\psi_{A,B\otimes C}^{FG} (FG(\id_A)\otimes \psi_{B,C}^{FG})= 
	\psi_{A\otimes B,C}^{FG}(\psi_{A,B}^{FG}\otimes FG(\id_C)).
\end{align*}
%
%
and that the unitality axioms are satisfied.
\end{proof}

As a further specialization, we also obtain a notion of \emph{colax cartesian lax-on-identities functor} as analogous to a colax cartesian functor (see Definition~\ref{def:(op) oplax cartesian functor}), but where the underling functor is lax monoidal lax-on-identities.

Now that we fixed the class of functors we deal with, we transfer the completeness result of Theorem~\ref{thm:simple_complete_bi_lax} to $\Rel$.
%

\begin{mytheorem}
	\label{thm:rel_complete_bi_lax}
	Let $\mC$ be a locally small oplax cartesian category $\mC$ and $f, g : X \to Y$ arrows in $\mC$. Then we have
	\begin{enumerate}
		\item $f\leq g$ if and only if $F(f)\leq F(g)$ for every colax cartesian lax-on-identities functor $F : \mC \to \Rel$;
		\item $f \approx g$ if and only if $F(f) \approx F(g)$ for every colax cartesian lax-on-identities functor $F : \mC \to \Rel$.
	\end{enumerate}
\end{mytheorem}

 \begin{proof}
 
 	To achieve this goal, we define a colax cartesian lax-on-identities functor
		$R : \Preord \rightarrow \Rel$
	that preserves and reflects the preorder on hom-sets, in the sense that for $f, g : X \to Y$ in $\Preord$, we have
	\begin{equation}
		\label{order_pres_refl}
		f \le g \quad \Longleftrightarrow \quad R(f) \subseteq R(g).
	\end{equation}
	Then post-composing any $F : \mC \to \Preord$ as in Theorem~\ref{thm:simple_complete_bi_lax} with $R$ results in a 
	colax cartesian lax-on-identities functor $R \circ F : \mC \to \Rel$ by Lemma~\ref{lem:compositin_lax_mon_and_lax_on_ident_mon}, and the completeness follows.
	
 	One way to construct such an $R : \Preord \rightarrow \Rel$ with the properties given above
	is to assign to every preordered set $(X,\leq_X)$ its underlying set $R(X,\leq_X):=X$ and to every monotone map $f$ its hypograph
	\[
		R(f) := \{ (x,y) \in X \times Y \mid y \le f(x) \}.
	\]
	Then condition~\eqref{order_pres_refl} is immediate.
	For $\freccia{X}{f}{Y}$ and $\freccia{Y}{g}{Z}$, strict functoriality
	\[
		R(g \circ f) = R(g) \circ R(f)
	\]
	is proved by showing both containments as follows. If $x \in X$ and $z \in Z$ are such that $x (R(g) \circ R(f)) z$, then this means that there is $y$ such that
	\[
		y \le f(x) \qquad z \le g(y)
	\]
	But then applying monotonicity of $g$ and transitivity of $\le$ implies $z \le g(f(x))$, which gives the desired $x R(g \circ f) z$.
	On the other hand, if $x R(g \circ f) z$ holds, then we can simply take $y \coloneqq f(x)$ as our witness of $x (R(g) \circ R(f)) z$, since then both of the above inequalities are satisfied.

	The lax preservation of identities $\id_X \subseteq R(\id_X)$ is easy to see, and it is worth noting that equality does not hold unless $X$ is discrete.

	Now we focus on the lax monoidal structure of $R$. A natural choice would be that of considering the strict monoidality of $R$ given by the trivial components $\psi_{A,B}=\id_{A\times B}$ and $\psi_I=\id_I$. However, in this case we would obtain that $R(!_A) = {}!_A$ but, in general, we would have just a strict inclusion $\nabla_A \subset R(\nabla_A)$, i.e.~$R$ would not be colax cartesian with the lax monoidal structure given by the trivial components. So we need to consider another lax monoidal structure for $R$. To this end we observe that, by definition of $R$, we have
	\begin{equation}\label{eq R(nabla)=R(id)nabla}
	 R(\nabla_A)=R(\id_{A\times A})\nabla_{RA}.
	\end{equation}
    Hence, we consider the lax monoidal structure of $R$ given by $\psi_{A,B} := R(\id_{A\times B})$ and $\psi_I := R(\id_I)=\id_I$. It is direct to see that such a choice of $\psi$ provides a natural transformation $\psi:\otimes \circ (R\times R)\to R\circ \otimes $, and that the associativity and unitality axioms of a lax monoidal lax-on-identities functor are satisfied.
    Thus, $R$ is a lax monoidal lax-on-identities functor.
    It is also colax cartesian since, by \eqref{eq R(nabla)=R(id)nabla}, we have
    \[
	    R(\nabla_A)=R(\id_{A\times A})\nabla_{R(A)}=\psi_{A,A}\nabla_{R(A)}
    \]
    and $R(!) =\psi_0 {}!$ holds since $\psi_0=\id_I$.
%
\end{proof}

\section{Conclusion and future works}

After a string-diagrammatic presentation of gs-monoidal categories and a few related structures.
our paper introduces their preorder-enriched extensions, and in particular oplax cartesian categories. 
We show that these categories represent a core language for relations and partial functions,
hence fitting in the current interest for the use of visual languages in modelling computational 
formalisms (see e.g.~\cite{Bonchi0Z21,fabio2022} and the references therein).
Such an interest in confirmed by our proof that such categories naturally arise in terms of Kleisli categories
and span categories, thus providing a large number of potential case studies.
%
As for the former, we also show that the canonical functor from the Kleisli category of a given monad
back to the original category has a gs-monoidal structure.
%
%
%
%
%
Finally, we turn to functorial completeness, showing it for oplax cartesian categories 
with respect to certain mappings into $\Rel$, thus generalising~\cite{CorradiniGadducci02}.
%

Future work will focus on completeness with respect to $\Rel$, in order to strengthen our result 
to genuine functors instead of lax-on-identities ones, following ideas in~\cite{Gadducci08}. 
We will also attempt to establish a stronger connection between the constructions involving Kleisli and span categories.
Furthermore, it could be interesting to address the traced monoidal case, as well as take into account the whole 2-categorical structure of $\mathbf{Span}(\mC)$,
for a presentation of graph rewriting and of inequational deduction for relational algebras, investigated in a 
set-theoretical flavour in~\cite{CorradiniGadducci99b,CorradiniGKK07},
and recently (at least for the diagrammatic presentation of the monoidal closed case) in~\cite{ghica}.

From a practical standpoint, we believe that our completeness results will find application in rewriting theory. The idea is that a completeness result makes it possible to derive new properties of a rewriting system in a straightforward manner, simply by analysing rewriting in a standard category like the one of preordered sets as in Theorem~\ref{thm:simple_complete_bi_lax}. Indeed, in rewriting the oplax structure amounts to changing the topology of the underlying graph \cite{CorradiniGadducci99b} so that the analysis of confluence could be performed simply by looking at the associated preorder in the model. Given that term graph rewriting is frequently used in functional programming language implementations, where the sharing of sub-terms is a key issue, our result may be concretely useful in improving such language implementations.


\section*{Declarations}

\bmhead{Author's Contribution} Tobias Fritz, Fabio Gadducci, Davide Trotta and Andrea Corradini contributed equally to this work.
\bmhead{Conflict of interest}The authors declare that they have no conflict of interest.
\bmhead{Availability of Data and Materials}Not applicable
\bmhead{Funding}Tobias Fritz acknowledges funding
by the Austrian Science Fund (FWF) through the project ``P 35992-N''. 
Andrea Corradini, Fabio Gadducci and Davide Trotta acknowledge funding
by the Italian Ministry of Education, University and Research (MIUR) through the project PRIN 20228KXFN2 ``STENDHAL''.

\bibliography{biblio_davide}


\begin{thebibliography}{51}
\ifx \bisbn   \undefined \def \bisbn  #1{ISBN #1}\fi
\ifx \binits  \undefined \def \binits#1{#1}\fi
\ifx \bauthor  \undefined \def \bauthor#1{#1}\fi
\ifx \batitle  \undefined \def \batitle#1{#1}\fi
\ifx \bjtitle  \undefined \def \bjtitle#1{#1}\fi
\ifx \bvolume  \undefined \def \bvolume#1{\textbf{#1}}\fi
\ifx \byear  \undefined \def \byear#1{#1}\fi
\ifx \bissue  \undefined \def \bissue#1{#1}\fi
\ifx \bfpage  \undefined \def \bfpage#1{#1}\fi
\ifx \blpage  \undefined \def \blpage #1{#1}\fi
\ifx \burl  \undefined \def \burl#1{\textsf{#1}}\fi
\ifx \doiurl  \undefined \def \doiurl#1{\url{https://doi.org/#1}}\fi
\ifx \betal  \undefined \def \betal{\textit{et al.}}\fi
\ifx \binstitute  \undefined \def \binstitute#1{#1}\fi
\ifx \binstitutionaled  \undefined \def \binstitutionaled#1{#1}\fi
\ifx \bctitle  \undefined \def \bctitle#1{#1}\fi
\ifx \beditor  \undefined \def \beditor#1{#1}\fi
\ifx \bpublisher  \undefined \def \bpublisher#1{#1}\fi
\ifx \bbtitle  \undefined \def \bbtitle#1{#1}\fi
\ifx \bedition  \undefined \def \bedition#1{#1}\fi
\ifx \bseriesno  \undefined \def \bseriesno#1{#1}\fi
\ifx \blocation  \undefined \def \blocation#1{#1}\fi
\ifx \bsertitle  \undefined \def \bsertitle#1{#1}\fi
\ifx \bsnm \undefined \def \bsnm#1{#1}\fi
\ifx \bsuffix \undefined \def \bsuffix#1{#1}\fi
\ifx \bparticle \undefined \def \bparticle#1{#1}\fi
\ifx \barticle \undefined \def \barticle#1{#1}\fi
\bibcommenthead
\ifx \bconfdate \undefined \def \bconfdate #1{#1}\fi
\ifx \botherref \undefined \def \botherref #1{#1}\fi
\ifx \url \undefined \def \url#1{\textsf{#1}}\fi
\ifx \bchapter \undefined \def \bchapter#1{#1}\fi
\ifx \bbook \undefined \def \bbook#1{#1}\fi
\ifx \bcomment \undefined \def \bcomment#1{#1}\fi
\ifx \oauthor \undefined \def \oauthor#1{#1}\fi
\ifx \citeauthoryear \undefined \def \citeauthoryear#1{#1}\fi
\ifx \endbibitem  \undefined \def \endbibitem {}\fi
\ifx \bconflocation  \undefined \def \bconflocation#1{#1}\fi
\ifx \arxivurl  \undefined \def \arxivurl#1{\textsf{#1}}\fi
\csname PreBibitemsHook\endcsname

\bibitem[\protect\citeauthoryear{Gadducci}{1996}]{gadducci1996}
\begin{botherref}
\oauthor{\bsnm{Gadducci}, \binits{F.}}:
On the algebraic approach to concurrent term rewriting.
PhD thesis,
University of Pisa
(1996)
\end{botherref}
\endbibitem

\bibitem[\protect\citeauthoryear{Corradini and
  Gadducci}{1997}]{CorradiniGadducci97}
\begin{bchapter}
\bauthor{\bsnm{Corradini}, \binits{A.}},
\bauthor{\bsnm{Gadducci}, \binits{F.}}:
\bctitle{A 2-categorical presentation of term graph rewriting}.
In: \beditor{\bsnm{Moggi}, \binits{E.}},
\beditor{\bsnm{Rosolini}, \binits{G.}} (eds.)
\bbtitle{CTCS 1997}.
\bsertitle{LNCS},
vol. \bseriesno{1290},
pp. \bfpage{87}--\blpage{105}.
\bpublisher{Springer},
\blocation{Berlin}
(\byear{1997})
\end{bchapter}
\endbibitem

\bibitem[\protect\citeauthoryear{Corradini and
  Gadducci}{1999a}]{CorradiniGadducci99}
\begin{barticle}
\bauthor{\bsnm{Corradini}, \binits{A.}},
\bauthor{\bsnm{Gadducci}, \binits{F.}}:
\batitle{An algebraic presentation of term graphs, via gs-monoidal categories}.
\bjtitle{Applied Categorical Structures}
\bvolume{7},
\bfpage{299}--\blpage{331}
(\byear{1999})
\end{barticle}
\endbibitem

\bibitem[\protect\citeauthoryear{Corradini and
  Gadducci}{1999b}]{CorradiniGadducci99b}
\begin{barticle}
\bauthor{\bsnm{Corradini}, \binits{A.}},
\bauthor{\bsnm{Gadducci}, \binits{F.}}:
\batitle{Rewriting on cyclic structures: equivalence between the operational
  and the categorical description}.
\bjtitle{RAIRO - Theoretical Informatics and Applications - Informatique
  Th\'eorique et Applications}
\bvolume{33}(\bissue{4-5}),
\bfpage{467}--\blpage{493}
(\byear{1999})
\end{barticle}
\endbibitem

\bibitem[\protect\citeauthoryear{Corradini and
  Gadducci}{2002}]{CorradiniGadducci02}
\begin{barticle}
\bauthor{\bsnm{Corradini}, \binits{A.}},
\bauthor{\bsnm{Gadducci}, \binits{F.}}:
\batitle{A functorial semantics for multi-algebras and partial algebras, with
  applications to syntax}.
\bjtitle{Theoretical Computer Science}
\bvolume{286},
\bfpage{293}--\blpage{322}
(\byear{2002})
\end{barticle}
\endbibitem

\bibitem[\protect\citeauthoryear{Fox}{1976}]{Fox:CACC}
\begin{barticle}
\bauthor{\bsnm{Fox}, \binits{T.}}:
\batitle{Coalgebras and cartesian categories}.
\bjtitle{Communications in Algebra}
\bvolume{4},
\bfpage{665}--\blpage{667}
(\byear{1976})
\end{barticle}
\endbibitem

\bibitem[\protect\citeauthoryear{Lawvere}{1963}]{Lawvere869}
\begin{barticle}
\bauthor{\bsnm{Lawvere}, \binits{F.W.}}:
\batitle{Functorial semantics of algebraic theories}.
\bjtitle{Proceedings of the National Academy of Sciences}
\bvolume{50}(\bissue{5}),
\bfpage{869}--\blpage{872}
(\byear{1963})
\end{barticle}
\endbibitem

\bibitem[\protect\citeauthoryear{Carboni and Walters}{1987}]{CARBONI198711}
\begin{barticle}
\bauthor{\bsnm{Carboni}, \binits{A.}},
\bauthor{\bsnm{Walters}, \binits{R.F.C.}}:
\batitle{Cartesian bicategories {I}}.
\bjtitle{Journal of Pure and Applied Algebra}
\bvolume{49}(\bissue{1}),
\bfpage{11}--\blpage{32}
(\byear{1987})
\end{barticle}
\endbibitem

\bibitem[\protect\citeauthoryear{Carboni et~al.}{2008}]{cartesianbicatII}
\begin{barticle}
\bauthor{\bsnm{Carboni}, \binits{A.}},
\bauthor{\bsnm{Kelly}, \binits{G.M.}},
\bauthor{\bsnm{Walters}, \binits{R.F.C.}},
\bauthor{\bsnm{Wood}, \binits{R.J.}}:
\batitle{Cartesian bicategories {II}}.
\bjtitle{Theory and Applications of Categories}
\bvolume{19}(\bissue{6}),
\bfpage{93}--\blpage{124}
(\byear{2008})
\end{barticle}
\endbibitem

\bibitem[\protect\citeauthoryear{Bonchi
  et~al.}{2018}]{bonchi_seeber_sobocinski_18}
\begin{bchapter}
\bauthor{\bsnm{Bonchi}, \binits{F.}},
\bauthor{\bsnm{Seeber}, \binits{J.}},
\bauthor{\bsnm{Soboci\'{n}ski}, \binits{P.}}:
\bctitle{Graphical conjunctive queries}.
In: \beditor{\bsnm{Ghica}, \binits{D.}},
\beditor{\bsnm{Jung}, \binits{A.}} (eds.)
\bbtitle{CSL 2018}.
\bsertitle{LIPIcs},
vol. \bseriesno{119},
pp. \bfpage{13}--\blpage{11323}.
\bpublisher{Schloss Dagstuhl--Leibniz-Zentrum f\"ur Informatik},
\blocation{Wadern}
(\byear{2018})
\end{bchapter}
\endbibitem

\bibitem[\protect\citeauthoryear{Bonchi et~al.}{2017}]{Bonchi2017c}
\begin{botherref}
\oauthor{\bsnm{Bonchi}, \binits{F.}},
\oauthor{\bsnm{Pavlovic}, \binits{D.}},
\oauthor{\bsnm{Soboci\'nski}, \binits{P.}}:
Functorial semantics for relational theories.
CoRR
\textbf{abs/1711.08699}
(2017)
\end{botherref}
\endbibitem

\bibitem[\protect\citeauthoryear{Fong and Spivak}{2019}]{Fong19}
\begin{botherref}
\oauthor{\bsnm{Fong}, \binits{B.}},
\oauthor{\bsnm{Spivak}, \binits{D.}}:
Regular and relational categories: Revisiting `{C}artesian bicategories {I}'.
CoRR
\textbf{abs/1909.00069}
(2019)
\end{botherref}
\endbibitem

\bibitem[\protect\citeauthoryear{Golubtsov}{1999}]{Gol99}
\begin{barticle}
\bauthor{\bsnm{Golubtsov}, \binits{P.V.}}:
\batitle{Axiomatic description of categories of information transformers}.
\bjtitle{Problems of Information Transmission}
\bvolume{35}(\bissue{3}),
\bfpage{259}--\blpage{274}
(\byear{1999})
\end{barticle}
\endbibitem

\bibitem[\protect\citeauthoryear{Cho and Jacobs}{2019}]{cho_jacobs_2019}
\begin{barticle}
\bauthor{\bsnm{Cho}, \binits{K.}},
\bauthor{\bsnm{Jacobs}, \binits{B.}}:
\batitle{Disintegration and {B}ayesian inversion via string diagrams}.
\bjtitle{Mathematical Structures in Computer Science}
\bvolume{29}(\bissue{7}),
\bfpage{938}--\blpage{971}
(\byear{2019})
\end{barticle}
\endbibitem

\bibitem[\protect\citeauthoryear{Fritz}{2020}]{Fritz_2020}
\begin{barticle}
\bauthor{\bsnm{Fritz}, \binits{T.}}:
\batitle{A synthetic approach to {M}arkov kernels, conditional independence and
  theorems on sufficient statistics}.
\bjtitle{Advances in Mathematics}
\bvolume{370},
\bfpage{107239}
(\byear{2020})
\end{barticle}
\endbibitem

\bibitem[\protect\citeauthoryear{Cockett and Lack}{2007}]{Cockett07}
\begin{barticle}
\bauthor{\bsnm{Cockett}, \binits{J.R.B.}},
\bauthor{\bsnm{Lack}, \binits{S.}}:
\batitle{Restriction categories {III}: colimits, partial limits and
  extensivity}.
\bjtitle{Mathematical Structures in Computer Science}
\bvolume{17}(\bissue{4}),
\bfpage{775}--\blpage{817}
(\byear{2007})
\end{barticle}
\endbibitem

\bibitem[\protect\citeauthoryear{Cockett and Lack}{2002}]{Cockett02}
\begin{barticle}
\bauthor{\bsnm{Cockett}, \binits{J.R.B.}},
\bauthor{\bsnm{Lack}, \binits{S.}}:
\batitle{Restriction categories {I}: categories of partial maps}.
\bjtitle{Theoretical Computer Science}
\bvolume{270}(\bissue{1}),
\bfpage{223}--\blpage{259}
(\byear{2002})
\end{barticle}
\endbibitem

\bibitem[\protect\citeauthoryear{Cockett and Lack}{2003}]{Cockett03}
\begin{barticle}
\bauthor{\bsnm{Cockett}, \binits{J.R.B.}},
\bauthor{\bsnm{Lack}, \binits{S.}}:
\batitle{Restriction categories {II}: partial map classification}.
\bjtitle{Theoretical Computer Science}
\bvolume{294}(\bissue{1}),
\bfpage{61}--\blpage{102}
(\byear{2003})
\end{barticle}
\endbibitem

\bibitem[\protect\citeauthoryear{Aguiar and Mahajan}{2010}]{aguiar2010}
\begin{bbook}
\bauthor{\bsnm{Aguiar}, \binits{M.}},
\bauthor{\bsnm{Mahajan}, \binits{S.}}:
\bbtitle{Monoidal Functors, Species and {H}opf Algebras}.
\bsertitle{CRM Monograph Series},
vol. \bseriesno{29},
p. \bfpage{784}.
\bpublisher{American Mathematical Society},
\blocation{Providence, RI}
(\byear{2010})
\end{bbook}
\endbibitem

\bibitem[\protect\citeauthoryear{Selinger}{2011}]{Selinger2011}
\begin{bchapter}
\bauthor{\bsnm{Selinger}, \binits{P.}}:
\bctitle{A survey of graphical languages for monoidal categories}.
In: \beditor{\bsnm{Coecke}, \binits{B.}} (ed.)
\bbtitle{New Structures for Physics}.
\bsertitle{Lecture Notes in Physics},
vol. \bseriesno{813},
pp. \bfpage{289}--\blpage{355}.
\bpublisher{Springer},
\blocation{Heidelberg}
(\byear{2011})
\end{bchapter}
\endbibitem

\bibitem[\protect\citeauthoryear{Fong and Spivak}{2019}]{fongspivak2019}
\begin{botherref}
\oauthor{\bsnm{Fong}, \binits{B.}},
\oauthor{\bsnm{Spivak}, \binits{D.}}:
Supplying bells and whistles in symmetric monoidal categories.
CoRR
\textbf{abs/1908.02633}
(2019)
\end{botherref}
\endbibitem

\bibitem[\protect\citeauthoryear{Robinson and Rosolini}{1988}]{Robinson88}
\begin{barticle}
\bauthor{\bsnm{Robinson}, \binits{E.}},
\bauthor{\bsnm{Rosolini}, \binits{G.}}:
\batitle{Categories of partial maps}.
\bjtitle{Information and Computation}
\bvolume{79}(\bissue{2}),
\bfpage{95}--\blpage{130}
(\byear{1988})
\end{barticle}
\endbibitem

\bibitem[\protect\citeauthoryear{Fritz et~al.}{2023}]{FritzGPR23}
\begin{barticle}
\bauthor{\bsnm{Fritz}, \binits{T.}},
\bauthor{\bsnm{Gonda}, \binits{T.}},
\bauthor{\bsnm{Perrone}, \binits{P.}},
\bauthor{\bsnm{Rischel}, \binits{E.F.}}:
\batitle{Representable {M}arkov categories and comparison of statistical
  experiments in categorical probability}.
\bjtitle{Theoretical Computer Science}
\bvolume{961},
\bfpage{113896}
(\byear{2023})
\end{barticle}
\endbibitem

\bibitem[\protect\citeauthoryear{Bruni and Gadducci}{2003}]{Bruni2003}
\begin{bchapter}
\bauthor{\bsnm{Bruni}, \binits{R.}},
\bauthor{\bsnm{Gadducci}, \binits{F.}}:
\bctitle{Some algebraic laws for spans (and their connections with
  multirelations)}.
In: \beditor{\bsnm{Kahl}, \binits{W.}},
\beditor{\bsnm{Parnas}, \binits{D.L.}},
\beditor{\bsnm{Schmidt}, \binits{G.}} (eds.)
\bbtitle{RelMiS 2001}.
\bsertitle{ENTCS},
vol. \bseriesno{44},
pp. \bfpage{175}--\blpage{193}.
\bpublisher{Elsevier},
\blocation{Amsterdam}
(\byear{2003})
\end{bchapter}
\endbibitem

\bibitem[\protect\citeauthoryear{Carboni}{1987}]{Carboni_87}
\begin{barticle}
\bauthor{\bsnm{Carboni}, \binits{A.}}:
\batitle{Bicategories of partial maps}.
\bjtitle{Cahiers de Topologie et G\'eom\'etrie Diff\'erentielle Cat\'egoriques}
\bvolume{28}(\bissue{2}),
\bfpage{111}--\blpage{126}
(\byear{1987})
\end{barticle}
\endbibitem

\bibitem[\protect\citeauthoryear{Kelly}{2005}]{Kelly05}
\begin{botherref}
\oauthor{\bsnm{Kelly}, \binits{G.M.}}:
Basic concepts of enriched category theory.
Reprints in Theory and Applications of Categories
(10),
1--136
(2005)
\end{botherref}
\endbibitem

\bibitem[\protect\citeauthoryear{Moggi}{1989}]{Moggi89}
\begin{bchapter}
\bauthor{\bsnm{Moggi}, \binits{E.}}:
\bctitle{Computational lambda-calculus and monads}.
In: \bbtitle{LICS 1989},
pp. \bfpage{14}--\blpage{23}.
\bpublisher{IEEE Press},
\blocation{CA, USA}
(\byear{1989})
\end{bchapter}
\endbibitem

\bibitem[\protect\citeauthoryear{Moggi}{1991}]{Moggi91}
\begin{barticle}
\bauthor{\bsnm{Moggi}, \binits{E.}}:
\batitle{Notions of computation and monads}.
\bjtitle{Information and Computation}
\bvolume{93}(\bissue{1}),
\bfpage{55}--\blpage{92}
(\byear{1991})
\end{barticle}
\endbibitem

\bibitem[\protect\citeauthoryear{Borceux}{1994}]{HCA2}
\begin{bbook}
\bauthor{\bsnm{Borceux}, \binits{F.}}:
\bbtitle{Handbook of Categorical Algebra 2: Categories and Structures}.
\bsertitle{Encyclopedia of Mathematics and its Applications},
vol. \bseriesno{51}.
\bpublisher{Cambridge University Press},
\blocation{Cambridge}
(\byear{1994})
\end{bbook}
\endbibitem

\bibitem[\protect\citeauthoryear{MacLane and Moerdijk}{1992}]{SGL}
\begin{bbook}
\bauthor{\bsnm{MacLane}, \binits{S.}},
\bauthor{\bsnm{Moerdijk}, \binits{I.}}:
\bbtitle{Sheaves in Geometry and Logic. A First Introduction to Topos Theory}.
\bpublisher{Springer},
\blocation{New York}
(\byear{1992})
\end{bbook}
\endbibitem

\bibitem[\protect\citeauthoryear{Makkai and Reyes}{1977}]{FOCL}
\begin{bbook}
\bauthor{\bsnm{Makkai}, \binits{M.}},
\bauthor{\bsnm{Reyes}, \binits{G.}}:
\bbtitle{First Order Categorical Logic}.
\bsertitle{Lecture Notes in Mathematics},
vol. \bseriesno{611}.
\bpublisher{Springer},
\blocation{Berlin-New York}
(\byear{1977})
\end{bbook}
\endbibitem

\bibitem[\protect\citeauthoryear{Lambek and Scott}{1986}]{IHOCL}
\begin{bbook}
\bauthor{\bsnm{Lambek}, \binits{J.}},
\bauthor{\bsnm{Scott}, \binits{P.J.}}:
\bbtitle{Introduction to Higher Order Categorical Logic}.
\bpublisher{Cambridge University Press},
\blocation{Cambridge}
(\byear{1986})
\end{bbook}
\endbibitem

\bibitem[\protect\citeauthoryear{Tanaka}{2004}]{PHDTT}
\begin{botherref}
\oauthor{\bsnm{Tanaka}, \binits{M.}}:
Pseudo-distributive laws and a unified framework for variable binding.
PhD thesis,
The University of Edinburgh
(2004)
\end{botherref}
\endbibitem

\bibitem[\protect\citeauthoryear{Giry}{1982}]{Giry82}
\begin{bchapter}
\bauthor{\bsnm{Giry}, \binits{M.}}:
\bctitle{A categorical approach to probability theory}.
In: \beditor{\bsnm{Banaschewski}, \binits{B.}} (ed.)
\bbtitle{Categorical Aspects of Topology and Analysis}.
\bsertitle{Lecture Notes in Mathematics},
vol. \bseriesno{915},
pp. \bfpage{68}--\blpage{85}.
\bpublisher{Springer},
\blocation{Berlin-New York}
(\byear{1982})
\end{bchapter}
\endbibitem

\bibitem[\protect\citeauthoryear{Jacobs}{2018}]{JACOBS2018}
\begin{barticle}
\bauthor{\bsnm{Jacobs}, \binits{B.}}:
\batitle{From probability monads to commutative effectuses}.
\bjtitle{Journal of Logical and Algebraic Methods in Programming}
\bvolume{94},
\bfpage{200}--\blpage{237}
(\byear{2018})
\end{barticle}
\endbibitem

\bibitem[\protect\citeauthoryear{Kock}{1970}]{Kock70}
\begin{barticle}
\bauthor{\bsnm{Kock}, \binits{A.}}:
\batitle{Monads on symmetric monoidal closed categories}.
\bjtitle{Archiv der Mathematik}
\bvolume{21},
\bfpage{1}--\blpage{10}
(\byear{1970})
\end{barticle}
\endbibitem

\bibitem[\protect\citeauthoryear{Kock}{1971}]{Kock71}
\begin{barticle}
\bauthor{\bsnm{Kock}, \binits{A.}}:
\batitle{Bilinearity and cartesian closed monads}.
\bjtitle{Mathematica Scandinavica}
\bvolume{29}(\bissue{2}),
\bfpage{161}--\blpage{174}
(\byear{1971})
\end{barticle}
\endbibitem

\bibitem[\protect\citeauthoryear{Jacobs}{1994}]{Jacobs1994}
\begin{barticle}
\bauthor{\bsnm{Jacobs}, \binits{B.}}:
\batitle{Semantics of weakening and contraction}.
\bjtitle{Annals of Pure and Applied Logic}
\bvolume{69}(\bissue{1}),
\bfpage{73}--\blpage{106}
(\byear{1994})
\end{barticle}
\endbibitem

\bibitem[\protect\citeauthoryear{Heunen et~al.}{2017}]{Heunen2017}
\begin{bchapter}
\bauthor{\bsnm{Heunen}, \binits{C.}},
\bauthor{\bsnm{Kammar}, \binits{O.}},
\bauthor{\bsnm{Staton}, \binits{S.}},
\bauthor{\bsnm{Yang}, \binits{H.}}:
\bctitle{A convenient category for higher-order probability theory}.
In: \bbtitle{LICS 2017},
p. \bfpage{12}.
\bpublisher{IEEE Press},
\blocation{Piscataway, NJ}
(\byear{2017})
\end{bchapter}
\endbibitem

\bibitem[\protect\citeauthoryear{Guitart}{1980}]{Guitart80}
\begin{barticle}
\bauthor{\bsnm{Guitart}, \binits{R.}}:
\batitle{Tenseurs et machines}.
\bjtitle{Cahiers de Topologie et G\'{e}om\'{e}trie Diff\'{e}rentielle}
\bvolume{21}(\bissue{1}),
\bfpage{5}--\blpage{62}
(\byear{1980})
\end{barticle}
\endbibitem

\bibitem[\protect\citeauthoryear{Fritz et~al.}{2021}]{FritzPR21}
\begin{barticle}
\bauthor{\bsnm{Fritz}, \binits{T.}},
\bauthor{\bsnm{Perrone}, \binits{P.}},
\bauthor{\bsnm{Rezagholi}, \binits{S.}}:
\batitle{Probability, valuations, hyperspace: Three monads on {T}op and the
  support as a morphism}.
\bjtitle{Math. Struct. Comput. Sci.}
\bvolume{31}(\bissue{8}),
\bfpage{850}--\blpage{897}
(\byear{2021})
\end{barticle}
\endbibitem

\bibitem[\protect\citeauthoryear{Day}{1970a}]{Day70b}
\begin{bchapter}
\bauthor{\bsnm{Day}, \binits{B.}}:
\bctitle{On closed categories of functors}.
In: \bbtitle{Reports of the {M}idwest {C}ategory {S}eminar, {IV}}.
\bsertitle{Lecture Notes in Mathematics},
vol. \bseriesno{137},
pp. \bfpage{1}--\blpage{38}.
\bpublisher{Springer},
\blocation{Berlin}
(\byear{1970})
\end{bchapter}
\endbibitem

\bibitem[\protect\citeauthoryear{Day}{1970b}]{Day70}
\begin{botherref}
\oauthor{\bsnm{Day}, \binits{B.J.}}:
Construction of biclosed categories.
PhD thesis,
University of New South Wales
(1970)
\end{botherref}
\endbibitem

\bibitem[\protect\citeauthoryear{Mandell et~al.}{2001}]{Mandell01}
\begin{barticle}
\bauthor{\bsnm{Mandell}, \binits{M.A.}},
\bauthor{\bsnm{May}, \binits{J.P.}},
\bauthor{\bsnm{Schwede}, \binits{S.}},
\bauthor{\bsnm{Shipley}, \binits{B.}}:
\batitle{Model categories of diagram spectra}.
\bjtitle{Proceedings of the London Mathematical Society}
\bvolume{82}(\bissue{2}),
\bfpage{441}--\blpage{512}
(\byear{2001})
\end{barticle}
\endbibitem

\bibitem[\protect\citeauthoryear{Bonchi et~al.}{2021}]{Bonchi0Z21}
\begin{bchapter}
\bauthor{\bsnm{Bonchi}, \binits{F.}},
\bauthor{\bsnm{Soboci\'{n}ski}, \binits{P.}},
\bauthor{\bsnm{Zanasi}, \binits{F.}}:
\bctitle{A survey of compositional signal flow theory}.
In: \beditor{\bsnm{Goedicke}, \binits{M.}},
\beditor{\bsnm{Neuhold}, \binits{E.J.}},
\beditor{\bsnm{Rannenberg}, \binits{K.}} (eds.)
\bbtitle{Advancing Research in Information and Communication Technology}.
\bsertitle{{IFIP} AICT},
vol. \bseriesno{600},
pp. \bfpage{29}--\blpage{56}.
\bpublisher{Springer},
\blocation{Switzerland}
(\byear{2021})
\end{bchapter}
\endbibitem

\bibitem[\protect\citeauthoryear{Bonchi et~al.}{2022}]{fabio2022}
\begin{barticle}
\bauthor{\bsnm{Bonchi}, \binits{F.}},
\bauthor{\bsnm{Gadducci}, \binits{F.}},
\bauthor{\bsnm{Kissinger}, \binits{A.}},
\bauthor{\bsnm{Soboci\'{n}ski}, \binits{P.}},
\bauthor{\bsnm{Zanasi}, \binits{F.}}:
\batitle{String diagram rewrite theory {I:} rewriting with {F}robenius
  structure}.
\bjtitle{Journal of the {ACM}}
\bvolume{69}(\bissue{2}),
\bfpage{14}--\blpage{11458}
(\byear{2022})
\end{barticle}
\endbibitem

\bibitem[\protect\citeauthoryear{Gadducci}{2008}]{Gadducci08}
\begin{bchapter}
\bauthor{\bsnm{Gadducci}, \binits{F.}}:
\bctitle{A term-graph syntax for algebras over multisets}.
In: \beditor{\bsnm{Corradini}, \binits{A.}},
\beditor{\bsnm{Montanari}, \binits{U.}} (eds.)
\bbtitle{WADT 2008}.
\bsertitle{LNCS},
vol. \bseriesno{5486},
pp. \bfpage{152}--\blpage{165}.
\bpublisher{Springer},
\blocation{Berlin, Heidelberg}
(\byear{2008})
\end{bchapter}
\endbibitem

\bibitem[\protect\citeauthoryear{Corradini et~al.}{2007}]{CorradiniGKK07}
\begin{bchapter}
\bauthor{\bsnm{Corradini}, \binits{A.}},
\bauthor{\bsnm{Gadducci}, \binits{F.}},
\bauthor{\bsnm{Kahl}, \binits{W.}},
\bauthor{\bsnm{K{\"{o}}nig}, \binits{B.}}:
\bctitle{Inequational deduction as term graph rewriting}.
In: \beditor{\bsnm{Mackie}, \binits{I.}},
\beditor{\bsnm{Plump}, \binits{D.}} (eds.)
\bbtitle{TERMGRAPH 2002}.
\bsertitle{ENTCS},
vol. \bseriesno{72},
pp. \bfpage{31}--\blpage{44}.
\bpublisher{Elsevier},
\blocation{Amsterdam}
(\byear{2007})
\end{bchapter}
\endbibitem

\bibitem[\protect\citeauthoryear{Alvarez{-}Picallo et~al.}{2022}]{ghica}
\begin{bchapter}
\bauthor{\bsnm{Alvarez{-}Picallo}, \binits{M.}},
\bauthor{\bsnm{Ghica}, \binits{D.}},
\bauthor{\bsnm{Sprunger}, \binits{D.}},
\bauthor{\bsnm{Zanasi}, \binits{F.}}:
\bctitle{Rewriting for monoidal closed categories}.
In: \beditor{\bsnm{Felty}, \binits{A.}} (ed.)
\bbtitle{FSCD 2022}.
\bsertitle{LIPIcs},
vol. \bseriesno{228},
pp. \bfpage{29}--\blpage{12920}.
\bpublisher{Schloss Dagstuhl - Leibniz-Zentrum f{\"{u}}r Informatik},
\blocation{Wadern}
(\byear{2022})
\end{bchapter}
\endbibitem

\bibitem[\protect\citeauthoryear{Kock}{1972}]{Kock72}
\begin{barticle}
\bauthor{\bsnm{Kock}, \binits{A.}}:
\batitle{Strong functors and monoidal monads}.
\bjtitle{Archiv der Mathematik}
\bvolume{23},
\bfpage{113}--\blpage{120}
(\byear{1972})
\end{barticle}
\endbibitem

\bibitem[\protect\citeauthoryear{Kock}{1970}]{Kock78}
\begin{barticle}
\bauthor{\bsnm{Kock}, \binits{A.}}:
\batitle{Monads on symmetric monoidal closed categories}.
\bjtitle{Archiv der Mathematik}
\bvolume{21},
\bfpage{1}--\blpage{10}
(\byear{1970})
\end{barticle}
\endbibitem

\end{thebibliography}

\begin{appendices}

\section{Lax/oplax/bilax monoidal functors}
\label{sec:lax_app}

This section recalls the definitions of lax, colax, and bilax monoidal functors, see e.g.~\cite{aguiar2010}.
Throughout, $\mC$ and $\mD$ are symmetric monoidal categories with tensor functor $\otimes$ and monoidal unit $I$, and we assume that 
$\otimes$ strictly associates without loss of generality in order to keep the diagrams simple.
Left and right unitors are denoted by $\lambda$ and $\rho$, respectively\footnote{Strict unitality could also be assumed, but that choice would make some diagrams potentially confusing.}, and braidings by $\gamma$.

\begin{mydefinition}\label{def:lax monoidal functor}
A functor $\freccia{\mC}{F}{\mD}$ is \textbf{lax monoidal} if it is equipped with a natural transformation 
\[\freccia{\otimes \circ \, (F\times F)}{\psi}{F\circ \otimes}\]
and an arrow $\freccia{I}{\psi_0}{F(I)}$ such that the associativity diagrams
\[\begin{tikzcd}[column sep=3ex]
	{F(A)\otimes F(B)\otimes F(C)} &&& {F(A)\otimes F(B\otimes C)} \\
	\\
	{F(A\otimes B)\otimes F(C)} &&& {F(A\otimes B\otimes C)}
	\arrow["{\id\otimes \,\psi_{B,C}}", from=1-1, to=1-4]
	\arrow["{\psi_{A,B}\otimes {}\id}"', from=1-1, to=3-1]
	\arrow["{\psi_{A\otimes B,C}}"', from=3-1, to=3-4]
	\arrow["{\psi_{A,B\otimes C}}", from=1-4, to=3-4]
\end{tikzcd}\]
and the unitality diagrams commute
\begin{equation}
\label{eq:lax_monoidal_unitality}
\begin{tikzcd}[column sep=1ex]
	{I\otimes F(A)} && F(A) && {F(A)\otimes  I} && F(A) \\
	\\
	{F(I)\otimes F(A)} && {F(I\otimes A)} && {F(A)\otimes F(I)} && {F(A\otimes  I).}
	\arrow["{\psi_{I,A}}"', from=3-1, to=3-3]
	\arrow["{F(\lambda_A)}", from=1-3, to=3-3]
	\arrow["{\psi_0\otimes {}\id}"', from=1-1, to=3-1]
	\arrow["{\lambda_{FA}}"', from=1-3, to=1-1]
	\arrow["{\rho_{FA}}"', from=1-7, to=1-5]
	\arrow["{\id\otimes \psi_0}"', from=1-5, to=3-5]
	\arrow["{\psi_{A,I}}"', from=3-5, to=3-7]
	\arrow["{F(\rho_A)}", from=1-7, to=3-7]
\end{tikzcd}
\end{equation}

$F$ is said to be \textbf{lax symmetric monoidal} if also the following diagram commutes
\[\begin{tikzcd}
	{F(A\otimes B)} && {F(B\otimes A)} \\
	{F(A)\otimes F(B)} && {F(B)\otimes F(A)}
	\arrow["{F(\gamma_{A,B})}", from=1-1, to=1-3]
	\arrow["{\psi_{A,B}}"', from=1-1, to=2-1]
	\arrow["{\gamma_{FA,FB}}"', from=2-1, to=2-3]
	\arrow["{\psi_ {B,A}}", from=1-3, to=2-3]
\end{tikzcd}\]
\end{mydefinition}

For example, if $\mC$ is the terminal monoidal category with only one object $I$ and $\id_I$ as the only arrow, then $F$ is simply a monoid in $\mD$.
We do not spell out the following dual version in full detail.

\begin{mydefinition}
	A functor $\freccia{\mC}{F}{\mD}$ is \textbf{oplax monoidal} if it is equipped with a
	natural transformation 
	\[\freccia{F\circ \otimes }{\phi}{ \otimes\circ (F\otimes F)}\]
	and a map
	$\freccia{F(I)}{\phi_0}{I}$
	satisfying axioms dual to those in Definition \ref{def:lax monoidal functor}. Similarly, an \textbf{oplax symmetric monoidal functor} is an oplax monoidal functor such that $\phi$ commutes with the braiding $\gamma$.
\end{mydefinition}

We also have the notion of \textbf{strong symmetric monoidal functor}, which is a lax symmetric monoidal functor with invertible structure arrows, or equivalently an oplax monoidal functor with invertible structure arrows; and that of \textbf{strict symmetric monoidal functor}, in which the structure arrows are identities.

A monoid and comonoid structure on an object in a symmetric monoidal category often interact in a nice way, either such that they form a \emph{bimonoid} or a \emph{Frobenius monoid} (and sometimes both).
The following definition (see~\cite{aguiar2010}) generalises the former notion to functors.

\begin{mydefinition}\label{def:bilax monoidalfunctor}
	A functor $\freccia{\mC}{F}{\mD}$ is \textbf{bilax monoidal} if it is equipped with a lax monoidal structure $\psi,\psi_0$ and an oplax monoidal structure $\phi,\phi_0$ such that the following compatibility conditions hold
\begin{itemize}
	\item 	\textbf{Braiding.} The following diagram commutes

	\[\begin{tikzcd}[column sep=-5ex]
		& {F(A\otimes B) \otimes F(C\otimes D)} \\
		{F(A\otimes B\otimes C\otimes D)} && {F(A)\otimes F(B)\otimes F(C)\otimes F(D)} \\
		{F(A\otimes C\otimes B\otimes D)} && {F(A)\otimes F(C)\otimes F(B)\otimes F(D)} \\
		& {F(A\otimes C)\otimes F(B\otimes D)}
		\arrow["{\psi_{A\otimes B,C\otimes D}}"', from=1-2, to=2-1]
		\arrow["{\phi_{A,B}\otimes \phi_{C,D}}", from=1-2, to=2-3]
		\arrow["{F(\id\otimes \gamma\otimes \id)}"', from=2-1, to=3-1]
		\arrow["{\id\otimes \gamma\otimes \id}", from=2-3, to=3-3]
		\arrow["{\phi_{A\otimes C,B\otimes D}}"', from=3-1, to=4-2]
		\arrow["{\psi_{A,C}\otimes \psi_{B,D}}", from=3-3, to=4-2]
	\end{tikzcd}\]
	\item \textbf{Unitality.} The following diagrams commute
\[\begin{tikzcd}
	I & {F(I)} & {F(I\otimes I)} & I & {F(I)} & {F(I\otimes I)} \\
	{I\otimes I} && {F(I)\otimes F(I)} & {I\otimes I} && {F(I) \otimes F(I)} \\
	&& {F(I)} \\
	& I && I
	\arrow["{\psi_0}", from=1-1, to=1-2]
	\arrow["{F(\lambda_I)}", from=1-2, to=1-3]
	\arrow["{\lambda_I}"', from=1-1, to=2-1]
	\arrow["{\phi_{I,I}}", from=1-3, to=2-3]
	\arrow["{\psi_{I,I}}"', from=2-6, to=1-6]
	\arrow["{F(\lambda_I^{-1})}"', from=1-6, to=1-5]
	\arrow["{\phi_0}"', from=1-5, to=1-4]
	\arrow["{\phi_0\otimes \phi_0}", from=2-6, to=2-4]
	\arrow["{\lambda_I^{-1}}", from=2-4, to=1-4]
	\arrow["{\psi_0}", from=4-2, to=3-3]
	\arrow["{\phi_0}", from=3-3, to=4-4]
	\arrow[Rightarrow, no head, from=4-2, to=4-4]
	\arrow["{\psi_0\otimes \psi_0}"', from=2-1, to=2-3]
\end{tikzcd}\]
\end{itemize}
\end{mydefinition}

We also say that $F$ is \textbf{bilax symmetric monoidal} if in addition both the lax and oplax structures are symmetric.

\section{Commutative monads}\label{sec: strong and commutative monad}
A \emph{strength} and a \emph{costrength} for a monad on a monoidal category are structures relating the monad with the tensor product of the category at least \emph{in one direction}. A monad equipped with a strength is called a strong monad.
This notion was introduced by Kock in \cite{Kock72,Kock78} as an alternative description of enriched monads.
Strong monads have been successfully used in computer science, playing a fundamental role in Moggi’s theory of computation \cite{Moggi89,Moggi91}.

We recall these concepts in the following definitions.
\begin{mydefinition}
A \textbf{strong monad} $(T,\mu,\eta,t)$ on a symmetric monoidal category $\mC$ is a monad $(T,\mu,\eta)$ on $\mC$ together with a natural transformation
\[
	\freccia{X\otimes T(Y)}{t_{X,Y}}{T(X\otimes Y)},
\]
called \textbf{strength}, such that the following diagrams commute for all objects $X$, $Y$, $Z$ of $\mC$


\[\begin{tikzcd}
	{X\otimes Y} & {X\otimes T(Y)}  & {I\otimes T(X)} & {T(I\otimes X)} \\
	& {T(X\otimes Y)} && {T(X)}\\
	{X\otimes Y\otimes T(Z)} && {T(X\otimes Y\otimes Z)} \\
	& {X\otimes T(Y\otimes Z)} \\
	{X \otimes T^2(Y)} & {T(X\otimes T(Y))} & {T^2(X\otimes Y)} \\
	{X\otimes T(Y)} && {T(X\otimes Y)}
	\arrow["{\lambda^{-1}_{T(X)}}"', from=1-3, to=2-4]
	\arrow["{t_{I,X}}", from=1-3, to=1-4]
	\arrow["{T(\lambda^{-1}_X)}", from=1-4, to=2-4]
	\arrow["{t_{X,T(Y)}}", from=5-1, to=5-2]
	\arrow["{T(t_{X,Y})}", from=5-2, to=5-3]
	\arrow["{\id_X\otimes \mu_Y}"', from=5-1, to=6-1]
	\arrow["{t_{X,Y}}"', from=6-1, to=6-3]
	\arrow["{\mu_{X\otimes Y}}", from=5-3, to=6-3]
	\arrow["{\id_X\otimes \eta_Y}", from=1-1, to=1-2]
	\arrow["{\eta_{X\otimes Y}}"', from=1-1, to=2-2]
	\arrow["{t_{X,Y}}", from=1-2, to=2-2]
	\arrow["{t_{X,Y\otimes Z}}"', from=4-2, to=3-3]
	\arrow["{\id_X\otimes t_{Y,Z}}"', from=3-1, to=4-2]
	\arrow["{t_{X\otimes Y,Z}}", from=3-1, to=3-3]
\end{tikzcd}\]
\end{mydefinition}

\begin{myexample}
	The list monad $\freccia{\Set}{T_{\mathrm{list}}}{\Set}$ is strong. Given two sets $X$ and $Y$, the strength component
	\[
		\freccia{X \times T_{\mathrm{list}}(Y)}{t_{X,Y}}{T_{\mathrm{list}}(X \times Y)}
	\]
	is given by the function assigning to an element $(x,[y_1,\dots,y_m])) $ of $ X\times T_{\mathrm{list}}(Y)$ the element $[(x,y_1),\dots,(x,y_m)]$ of $ T_{\mathrm{list}}(X\times Y)$.
	
	In fact, \emph{any} monad on the cartesian category $\Set$ is strong in a unique way, where the strength can be defined similarly to the strength of the list monad. We refer to  \cite{Kock72,Kock78} for more details.
\end{myexample}

\begin{myremark}
The braiding $\gamma$ of $\mC$ let us define a \textbf{costrength} with components
\[\freccia{T(X)\otimes Y}{t_{X,Y}'}{T(X\otimes Y)}\]
given by
\[t_{X,Y}':=T(\gamma_{Y,X}) \circ t_{Y,X} \circ \gamma_{T(X),Y}.\]
It satisfies axioms that are analogous to those of strength.
\end{myremark}

\begin{mydefinition}
\label{commutative_monad}
A strong monad $(T,\mu,\eta,t)$ on a symmetric monoidal category $\mC$ is said to be \textbf{commutative} if the following 
diagram commutes for every object $X$ and $Y$
\begin{equation}
	\label{commutative_monad_diag}
\begin{tikzcd}
	{T(X)\otimes T(Y)} & {T(T(X)\otimes Y)} & {T^2(X\otimes Y)} \\
	{T(X\otimes T(Y))} & {T^2(X\otimes Y)} & {T(X\otimes Y)}
	\arrow["{t_{T(X),Y}}", from=1-1, to=1-2]
	\arrow["{T(t'_{X,Y})}", from=1-2, to=1-3]
	\arrow["{t'_{X,T(Y)}}"', from=1-1, to=2-1]
	\arrow["{T(t_{X,Y})}"', from=2-1, to=2-2]
	\arrow["{\mu_{X\otimes Y}}"', from=2-2, to=2-3]
	\arrow["{\mu_{X\otimes Y}}", from=1-3, to=2-3]
\end{tikzcd}
\end{equation}

\end{mydefinition}

\begin{myremark}
	\label{rem:comm_vs_symmon}
	It is well-known that on a symmetric monoidal category, commutative monads are equivalent to symmetric monoidal monads~\cite[Theorem~2.3]{Kock72}.
	Indeed, the diagonal of~\eqref{commutative_monad_diag} equips the functor $T$ with a lax symmetric monoidal structure, whose components we denote by $c_{X,Y} : T(X) \otimes T(Y) \to T(X \otimes Y)$.
\end{myremark}

\end{appendices}


\end{document}